\documentclass[11pt]{article}

\pdfoutput=1

\usepackage{style}
\usepackage{shortcuts}

\usepackage{my-diagrams}
\usepackage{young-tableaux}

\usepackage{calc}
\usepackage[boxsize=.65em, centertableaux]{ytableau}

\title{The Computational Complexity of Plethysm Coefficients}
\author{%
	Nick Fischer%
	\thanks{\email{nfischer@mpi-inf.mpg.de}}
	\institute{Max Planck Institute for Informatics,\\Saarbr\"ucken, Germany}
	\institute{Saarbr\"ucken Graduate School\\of Computer Science}
	\and
	Christian Ikenmeyer%
	\thanks{\email{christian.ikenmeyer@liverpool.ac.uk} This research was partially carried out when the author was at the Max Planck Institute for Software Systems, Saarbr\"ucken. The author acknowledges support from the DFG grant IK 116/2-1.}
	\institute{The University of Liverpool}}

\begin{document}

\maketitle
\thispagestyle{empty}

\begin{abstract}
In two papers, B\"urgisser and Ikenmeyer (STOC 2011, STOC 2013) used an adaption of the geometric complexity theory (GCT) approach by Mulmuley and Sohoni (Siam J Comput 2001, 2008) to prove lower bounds on the border rank of the matrix multiplication tensor. A key ingredient was information about certain Kronecker coefficients. While tensors are an interesting test bed for GCT ideas, the far-away goal is the separation of algebraic complexity classes. The role of the Kronecker coefficients in that setting is taken by the so-called plethysm coefficients: These are the multiplicities in the coordinate rings of spaces of polynomials. Even though several hardness results for Kronecker coefficients are known, there are almost no results about the complexity of computing the plethysm coefficients or even deciding their positivity.

In this paper we show that deciding positivity of plethysm coefficients is \NP{}-hard, and that computing plethysm coefficients is \SharpP{}-hard. In fact, both problems remain hard even if the inner parameter of the plethysm coefficient is fixed. In this way we obtain an inner versus outer contrast: If the outer parameter of the plethysm coefficient is fixed, then the plethysm coefficient can be computed in polynomial time.

Moreover, we derive new lower and upper bounds and in special cases even combinatorial descriptions for plethysm coefficients, which we consider to be of independent interest. Our technique uses discrete tomography in a more refined way than the recent work on Kronecker coefficients by Ikenmeyer, Mulmuley, and Walter (Comput Compl 2017). This makes our work the first to apply techniques from discrete tomography to the study of plethysm coefficients. Quite surprisingly, that interpretation also leads to new equalities between certain plethysm coefficients and Kronecker coefficients.
\end{abstract}

\pagestyle{plain}
\newpage
\setcounter{page}{1}

% !TEX root = ../paper.tex

\section{Introduction} \label{sec:introduction}
Geometric complexity theory (GCT) is an approach towards the separation of Valiant's algebraic complexity classes using algebraic geometry and representation theory. These ideas were introduced by Mulmuley and Sohoni in~\cite{MulmuleyS2001,MulmuleyS2008}. A first implementation of that framework proves lower bounds on the border rank of the matrix multiplication tensor \cite{BurgisserI2011,BurgisserI2013b}. In these papers, \emph{Kronecker coefficients} play a key role, as they are the multiplicities in the coordinate rings of spaces of tensors. While that line of research turns out as an interesting test case for GCT ideas, the ultimate goal remains to separate algebraic complexity classes, that is, certain sets of (families of) polynomials. Switching from tensors to polynomials, the role of Kronecker coefficients in \cite{BurgisserI2011,BurgisserI2013b} is now taken by \emph{plethysm coefficients} (see e.g.~\cite[Sec.~12.4(i)]{BlaserI2018}). Plethysm coefficients are the main subject of ``GCT7'' \cite{Mulmuley2007}, and they also appear prominently in the GCT publications \cite{KadishL2012, Kumar2015, Burgisser2016, BurgisserHI2017}.

The main subject of study in this paper are plethysm coefficients and, to a lesser extent, Kronecker coefficients.
Both types of representation theoretic coefficients are not only important in GCT, but they are fundamental objects in algebraic combinatorics. Indeed, in 1999, Richard Stanley highlighted the quest for a combinatorial description for plethysm coefficients as Problem 9 in his survey on ``outstanding open problems'' in algebraic combinatorics~\cite{Stanley1999a}; finding a combinatorial interpretation for Kronecker coefficients is Problem 10 in the same paper\footnote{In the language of computational complexity theory, ``finding a combinatorial description'' is commonly interpreted as ``proving that the function that computes the coefficient from its parameter list is in the complexity class~\SharpP''. \cite{Mulmuley2007} lists conjectures that would imply the containment in \SharpP, thus proving Stanley's Problem~9.}.
Even though GCT asks to compare these (and related) coefficients (see e.g.\ \cite[Appendix]{Ikenmeyer2012}), both coefficients have been studied mostly independently. But this has started to change recently, as can be seen for example from the fact that the Fall 2016 Eastern Sectional Meeting of the American Mathematical Society held a special session on ``Plethysm and Kronecker products in Representation Theory'' and also from the 2020 Oberwolfach Mini-Workshop on ``Kronecker, Plethysm, and Sylow Branching Coefficients and their applications to Complexity Theory''.
To the best of our knowledge the only papers that give inequalities between these two sorts of coefficients are \cite{Manivel2011, IkenmeyerP2017,BurgisserIP2019}\footnote{\cite{Mulmuley2009} writes that ``The Kronecker coefficients [...] are [...] special cases of the fundamental plethysm constants'', however, that definition of ``plethysm constants'' is much broader than the plethysm coefficients that we study in this paper.}.

In GCT, both coefficients do not only serve as the multiplicities in the coordinate rings of spaces of tensors or polynomials, but they also appear as terms in nonnegative formulas that describe multiplicities of coordinate rings of important group orbits: The plethysm coefficients appear as terms in the nonnegative formulas for the multiplicities in the coordinate ring of the orbit of the product of variables (\cite[Sec.~9.2.3]{Landsberg2017}), the power sum polynomial (\cite{IkenmeyerK2019}), the unit tensor (\cite[Sec.~3.7]{Ikenmeyer2019}), and the permanent polynomial (\cite[Eq.~(5.5.2)]{BurgisserLMW2011}), whereas the Kronecker coefficients appear as terms in the nonnegative formulas for the multiplicities in the coordinate ring of the orbit of the determinant polynomial \cite[Eq.~(5.2.6)]{BurgisserLMW2011}, the permanent polynomial (\cite[Eq.~(5.5.2)]{BurgisserLMW2011}, both the plethysm coefficients and the Kronecker coefficients appear in this formula), the matrix multiplication tensor \cite[Thm.~5.3]{BurgisserI2011}, the unit tensor (\cite[Sec.~3.7]{Ikenmeyer2019}, both the plethysm coefficients and the Kronecker coefficients appear in this formula), and the trace of the matrix power polynomial \cite[Thm~2.10]{GesmundoIP2017}. In theoretical physics, the plethysm and Kronecker coefficients are closely related to the quantum marginal problem for indistinguishable and distinguishable particles, respectively, see e.g.~\cite{ChristiandlHM2007} and \cite[Sec.~7]{ChristandlDKW2014}, where the asymptotic positivity of the coefficients is studied.

While the \SharpP{}-hardness of the computation of the Kronecker coefficient dates back to \cite{BurgisserI2008,BriandOR2009}, the recent paper \cite{IkenmeyerMW2017} proves the \NP{}-hardness of deciding the positivity of the Kronecker coefficient. This was a setback for the GCT program, as it was originally conjectured that Kronecker positivity would be decidable in polynomial time, much in the same way as the well-known Littlewood-Richardson coefficients: Even though the computation of Littlewood-Richardson coefficients is $\SharpP$-complete~\cite{Narayanan2006}, positivity of Littlewood-Richardson coefficients can be decided in polynomial time~\cite{DeLoeraM2006} using linear programming (see also~\cite{MulmuleyNS2012}), and even with a combinatorial max-flow algorithm~\cite{BurgisserI2013a}. Such a result would have made the search for obstructions (i.e., inequalities between coefficients) much easier.

In this paper we prove \NP{}-hardness of deciding positivity of plethysm coefficients~(\cref{thm:plethysm-positivity-np-hard}) and \SharpP{}-hardness of the computation of plethysm coefficients~(\cref{thm:plethysm-sharpp-hard}). Not even the \SharpP{}-hardness of the computation of plethysm coefficients was known prior to our work. Indeed, our reduction is quite subtle compared to the one used to prove \NP-hardness of deciding positivity of Kronecker coefficients \cite{IkenmeyerMW2017}.

\subsection*{Structure of the Paper}
\cref{sec:preliminaries} explains the necessary background from representation theory. \cref{sec:results} states our main hardness results and puts them in contrast to what is known about efficiently computable subcases. \cref{sec:m=3} is a first important step in our reduction: We show that the inner parameter of the plethysm coefficient can be fixed to be three. \cref{sec:plethysm-bounds} translates an interesting subcase of the plethysm coefficient problem into a problem in discrete tomography. In \cref{sec:reduction} we use this tomography description to design a sequence of reductions that finally prove our hardness result. \cref{sec:kronecker,sec:positive-formulas} contain results of independent interest: \cref{sec:kronecker} highlights new close connections to Kronecker coefficients, while \cref{sec:positive-formulas} contains more combinatorial descriptions for the plethysm coefficients. Finally, \cref{sec:gapp-completeness} uses classical results to provide an algorithm that places the problem of computing plethysm coefficients in the complexity class~\GapP{}. 

% !TEX root = ../paper.tex

\section{Preliminaries from Representation Theory} \label{sec:preliminaries}
The definition of plethysm and Kronecker coefficients that we use requires some elementary algebraic geometry and representation theory, as can be found for example in the standard textbooks \cite{Kraft1985,FultonH1991,Sagan2001}. The necessary material can also for example be found in \cite[Ch.~6]{Landsberg2011}.

A \emph{composition $\lambda$} is a list of nonnegative integers $\lambda = (\lambda_0, \lambda_1, \ldots)$, which we always treat as finite by omitting trailing zeros. If the entries are nonincreasing, then we call $\lambda$ a \emph{partition}. A partition~$\lambda$ is often presented as its \emph{Young diagram}, which is a top-left justified array containing~$\lambda_i$ boxes in the $i$-th row. For instance, $(3, 1)$ is identified with the Young diagram $\ydiagram{3,1}$. In that presentation it makes sense to term $\lambda_0$ the \emph{width} of $\lambda$ and the number of non-zero entries the \emph{height} of $\lambda$. We refer to the number $|\lambda| := \sum_i \lambda_i$ of boxes in $\lambda$'s Young diagram as the \emph{size} of~$\lambda$. If $\lambda$ is of height $1$, then we call $\lambda$ a \emph{1-row partition}, and similarly, if $\lambda$ is of width $1$, then $\lambda$ is a \emph{1-column partition}. Finally, let the \emph{transpose $\lambda^t$} of $\lambda$ denote the partition obtained by exchanging rows and columns. For example, $(3, 1)^\transpose = (2, 1, 1)$ is depicted as \raisebox{-.8ex}[2ex][2.9ex]{$\ydiagram{2,1,1}$}.

Let $V$ be a finite-dimensional complex vector space. There is a canonical action of the general linear group $\GL(V)$ on $V$ given by $gv := g(v)$ for all $g \in \GL(V)$, $v \in V$. For a composition $\lambda$, if $\diag(\alpha_0, \ldots, \alpha_k) v = \alpha_0^{\lambda_0} \cdots \alpha_k^{\lambda_k} v$, then we say that $v$ is a \emph{weight vector of weight $\lambda$}. Note that for all matrices $m \in \gl(V) := \Complex^{\dim V \times \dim V}$ we have that $\Id + \varepsilon m \in \GL(V)$ for all small enough~$\varepsilon$, where $\Id$ is the $\dim V \times \dim V$ identity matrix. We say that $V$ admits an action of the Lie-algebra~$\gl(V)$ defined by $m v := \lim_{\varepsilon \to 0} \varepsilon^{-1} ((\Id + \varepsilon m)v - v)$ for all $m \in \gl(V)$, $v \in V$.
For $i < j$, define the \emph{raising operator} $E_{i,j} \in \gl(V)$ as the matrix with a single one at position $(i, j)$ and zeros everywhere else. A weight vector $v$ of weight $\lambda$ is called a \emph{highest weight vector} if $v$ vanishes under the action of all raising operators. In that case, $\lambda$ is guaranteed to be a partition. It is well-known that the irreducible polynomial representations of $\GL(V)$ are characterized by their unique (up to scale) highest weight vectors, which in turn, are indexed by partitions of height at most $\dim V$. The irreducible $\GL(V)$-representation of type $\nu$ is called the \emph{Weyl module} and denoted by~$\Schur^\nu V$. The Weyl module $\Schur^\nu V$ is constructed as follows. Given any $\GL(V)$-representation $W$ (for example $W = V$), there is a natural linear $\GL(V)$-action on its $m$-th tensor power $\Tensor^m W$ given by $g(w_1 \tensor w_2 \tensor \cdots \tensor w_m) := (g w_1) \tensor (g w_2) \tensor \cdots \tensor (g w_m)$ and linear continuation. We define $\Schur^\nu W$ to be the image of the \emph{Schur functor $\Schur^\nu$} that maps $W$ to its $m$-th tensor power~$\Tensor^m W$ and then projects onto the $\nu$-th isotypic component via
\[
	w \mapsto \tfrac 1 {m!} \sum_{\pi \in \S_m} \chi_\nu(\pi) \pi w,
\]
where $\chi_\nu$ is the character of the irreducible representation of type $\nu$ of the symmetric group $\S_m$. For example, if $\nu=(m)$, then $S^\nu V$ is the vector space of symmetric tensors of order $m$, which we denote by $\Sym^m V$. If $\nu=(1,1,\ldots,1)$ is of size $m$, then $S^\nu V$ is the vector space of alternating tensors of order $m$, which we denote by $\Wedge^m V$.

A \emph{plethysm} is the application of $\Schur^\mu$ to $\Schur^\nu V$. The discussion above shows that this space $\Schur^\mu \Schur^\nu V := \Schur^\mu(\Schur^\nu V)$ is a $\GL(V)$-representation. Since $\GL(V)$ is reductive, every finite-dimensional $\GL(V)$-representation decomposes into a direct sum of irreducible $\GL(V)$-representations:
\begin{align*}
  \Schur^\mu \Schur^\nu V = \Directsum_\lambda (\Schur^\lambda V)^{\directsum p_\lambda(\mu, \nu)}.
\end{align*}
The nonnegative integers $p_\lambda(\mu, \nu)$ are called \emph{general plethysm coefficients}. The case where $\mu=(n)$ and $\nu=(m)$ is of special interest, so we write $a_\lambda(n, m) := p_\lambda((n), (m))$, which is the multiplicity of type $\lambda$ in $\Sym^n \Sym^m V$. The coefficient $a_\lambda(n, m)$ is often called the \emph{plethysm coefficient} in the literature.
We introduce the \emph{dual plethysm coefficient $b_\lambda(n, m)$} as another special case of general plethysm coefficients, defined by
\begin{align}
	\Sym^n \Sym^m V = \Directsum_\lambda (\Schur^\lambda V)^{\directsum a_\lambda(n, m)} &&
	\text{and} &&
	\Wedge^n \Sym^m V = \Directsum_\lambda (\Schur^\lambda V)^{\directsum b_\lambda(n, m)}, \label{eq:decomp-sym-sym}
\end{align}
i.e., by restricting the Young diagrams $\mu$ and $\nu$ to a single row (in case of $\Sym$) or a single column (in case of $\Wedge$). By some well-known dualities~\cite{Carre1990, Manivel1998, ManivelM2015} we can relate the above decompositions with their respective analogues after exchanging $\Sym$ and $\Wedge$.

\begin{fact}[\cite{Carre1990, Manivel1998, ManivelM2015}] \label[fact]{fact:duality}
If $m$ is odd, then:
\begin{align}
	\Wedge^n \Wedge^m V = \Directsum_\lambda (\Schur^\lambda V)^{\directsum a_{\lambda^\transpose}(n, m)} &&
	\text{and} &&
	\Sym^n \Wedge^m V = \Directsum_\lambda (\Schur^\lambda V)^{\directsum b_{\lambda^\transpose}(n, m)}, \label{eq:decomp-wedge-wedge-odd}
\end{align}
whereas in case that $m$ is even:
\begin{align}
	\Wedge^n \Wedge^m V = \Directsum_\lambda (\Schur^\lambda V)^{\directsum b_{\lambda^\transpose}(n, m)} &&
	\text{and} &&
	\Sym^n \Wedge^m V = \Directsum_\lambda (\Schur^\lambda V)^{\directsum a_{\lambda^\transpose}(n, m)}. \label{eq:decomp-wedge-wedge-even}
\end{align}
(Notice that in these decompositions \eqref{eq:decomp-wedge-wedge-odd} and \eqref{eq:decomp-wedge-wedge-even}, $\lambda$ occurs transposed in $a_{\lambda^\transpose}(n, m), b_{\lambda^\transpose}(n, m)$ as opposed to \eqref{eq:decomp-sym-sym}.)
\end{fact}

% !TEX root = ../paper.tex

\section{Main Results: Complexity of Plethysm Coefficients} \label{sec:results}
In order to establish hardness of computing general plethysm coefficients $p_\lambda(\mu, \nu)$, it is of course sufficient to focus on the interesting special cases $a_\lambda(n, m)$ and $b_\lambda(n, m)$. To that end, we study the following problems from a complexity theoretic perspective\footnotemark.
\footnotetext{In all problems the inputs can be encoded in binary or in unary, as it makes no difference for our results below. The hardness results still hold for the unary encoding, while the polynomial-time algorithms still work for the binary encoding.}

\begin{problem}[\Plethysm]
Given a partition $\lambda$ and two integers $n$ and $m$, output $a_\lambda(n,m)$.
\end{problem}

\begin{problem}[\PlethysmPositivity]
Given a partition $\lambda$ and two integers $n$ and $m$, output ``accept'' if $a_\lambda(n,m)>0$, and output ``reject'' otherwise.
\end{problem}

\begin{problem}[\DualPlethysm]
Given a partition $\lambda$ and two integers $n$ and $m$, output $b_\lambda(n,m)$.
\end{problem}

\begin{problem}[\DualPlethysmPositivity]
Given a partition $\lambda$ and two integers $n$ and $m$, output ``accept'' if $b_\lambda(n,m)>0$, and output ``reject'' otherwise.
\end{problem}

We remark that the computation of $a_\lambda(n, 2)$ and $b_\lambda(n, 2)$ is possible in polynomial time\footnotemark, so we need to deal with cases $m \geq 3$ only. We will be particularly interested in cases where one of the parameters $m$ or $n$ is fixed. For each of the problems above we therefore define versions where $m$, respectively $n$, is fixed by appending \Inner{$m$}, respectively \Outer{$n$}, to the problem name. For example, \Plethysm\Inner{$3$} is the problem ``given $\lambda$ and $n$, compute $a_\lambda(n,3)$''.
\footnotetext{By~\cite[Proposition~2.3.8, Proposition~2.3.9]{Stanley1999b, Weyman2003}, $\Sym^n \Sym^2 V$ decomposes multiplicity-free as \raisebox{0pt}[0pt][0pt]{$\Directsum_{\lambda \in A(n)} \Schur^\lambda V$} for $A(n) = \{ \text{$\lambda$ of size $2n$} : \text{each $\lambda_i$ is even} \}$; similarly, we have \raisebox{0pt}[0pt][0pt]{$\Wedge^n \Sym^2 V = \Directsum_{\lambda \in B(n)} \Schur^\lambda V$}, where $B(n) = \{ \text{$\lambda$ of size $2n$} : \text{if $\lambda_i \geq i$, then $\lambda_i = \lambda^\transpose_i + 1$} \}$. Testing membership $\lambda \in A_n$ and $\lambda \in B_n$ can both be performed in polynomial time in~$|\lambda|$.}

We remark that $\PlethysmPositivity$ is called the ``zero locus of plethysm coefficients'' in~\cite{KahleM2016}, where this problem is labeled as ``highly nontrivial''. Our main result makes this intuitive statement precise. We prove that $\PlethysmPositivity$ is $\NP$-hard:

\begin{theorem} \label[theorem]{thm:plethysm-positivity-np-hard}
\Plethysm\Inner{$m$} and \DualPlethysm\Inner{$m$} are \NP{}-hard for any fixed $m \geq 3$. In particular, \PlethysmPositivity{} and \DualPlethysmPositivity{} are \NP{}-hard.
\end{theorem}

By a thorough analysis of the proof, show that $\Plethysm$ is \SharpP{}-hard:

\begin{theorem} \label[theorem]{thm:plethysm-sharpp-hard}
\Plethysm\Inner{$m$} and \DualPlethysm\Inner{$m$} are \SharpP{}-hard for any fixed $m \geq 3$. In particular, \Plethysm{} and \DualPlethysm{} are \SharpP{}-hard.
\end{theorem}

Considering \cref{thm:plethysm-positivity-np-hard} and \cref{thm:plethysm-sharpp-hard}, we obtain an interesting complexity-theoretic contrast if we fix the outer parameter instead of the inner, as seen in the following proposition, whose proof can be extracted from \cite{KahleM2016}\footnotemark.
\footnotetext{\cite{KahleM2016} shows that, for any fixed partition $\mu$, the function $(\lambda_0, \ldots, \lambda_{|\mu|-1}) \mapsto p_\lambda(\mu, (m))$ has a constant size arithmetic formula (with modular arithmetic) whose inputs are the $\lambda_i$; see the appendix of the arXiv version~\cite{KahleM2015} for a good exposition.}

\begin{proposition} \label[proposition]{prop:plethysm-dichotomy}
\Plethysm\Outer{$n$} and \DualPlethysm\Outer{$n$} are in \P{} for any fixed~$n$. More generally, computing $p_\lambda(\mu, (m))$ can be done in polynomial time if $|\mu|$ is fixed.
\end{proposition}

This striking difference of complexities between \cref{thm:plethysm-positivity-np-hard} (if $m$ is fixed, then even deciding positivity of $a_\lambda(n,m)$ is $\NP$-hard) and \cref{prop:plethysm-dichotomy} (if $n$ is fixed, then even computing the exact value of $a_\lambda(n,m)$ can be done in polynomial time) could be interpreted as an explanation why it is considered to be much more difficult to obtain exact results about $a_\lambda(n,m)$ for $n \gg m$ compared to the case where $m \gg n$, see e.g.~\cite{Weintraub1990, Thibon1992, KahleM2016, KahleM2018, KimotoS2019} for evidence.

In addition, the \SharpP{}-hardness result leads to a completeness result for the complexity class \GapP{}---that is, the class of all counting problems implementable by a nondeterministic polynomial-time Turing machine, where the output is interpreted as the number accepting computation paths minus the number of rejecting computation paths~\cite{FennerFK1994}.

\begin{theorem} \label{thm:plethysm-gapp-complete}
\Plethysm{}, \DualPlethysm{} and the problem of computing general plethysm coefficients are \GapP{}-complete.
\end{theorem}

The proof of~\cref{thm:plethysm-gapp-complete} is provided in~\cref{sec:gapp-completeness}.

% !TEX root = ../paper.tex

\section{Fixing the Inner Parameter to \boldmath$m = 3$} \label{sec:m=3}

To ultimately prove \cref{thm:plethysm-positivity-np-hard,thm:plethysm-sharpp-hard}, we preliminarily demonstrate that proving hardness of the case $m = 3$ is enough:

\begin{lemma} \label[lemma]{lem:reduction-m=3}
Let $m$ be fixed. Given a partition $\lambda$ encoded in unary, we can compute partitions~$\pi$ and~$\pi'$ in polynomial time so that $a_\lambda(n, m) = b_\pi(n, m + 1)$ and $b_\lambda(n, m) = a_{\pi'}(n, m + 1)$.
\end{lemma}

\noindent
In particular, provided that deciding positivity of both $a_\lambda(n, 3)$ and $b_\lambda(n, 3)$ is already \NP{}-hard, deciding positivity of a plethysm coefficient $a_\lambda(n, m)$ or $b_\lambda(n, m)$ is \NP{}-hard as well for each fixed $m \geq 3$. In order to prove \cref{lem:reduction-m=3}, we leverage another fact:

\begin{fact}[\cite{ManivelM2015}] \label[fact]{fact:reduction-wedge}
Let $\nu$ be a partition of~$n$, let~$\lambda$ be a partition of $n \mult m$ of width at most~$n$ and let~$\pi$ be the partition after prepending a row of width~$n$ to~$\lambda$. Then the multiplicity of the irreducible component $\Schur^\lambda V$ in $\Schur^\nu \Wedge^m V$ equals the multiplicity of the irreducible component $\Schur^\pi V$ in $\Schur^\nu \Wedge^{m + 1} V$.
\end{fact}

\begin{proof}[Proof of \cref{lem:reduction-m=3}]
We show the first claim only; the second part is analogous. Let a partition~$\lambda$ of~$n \mult m$ be given. \cref{fact:duality} yields that $a_\lambda(n, m)$ equals the multiplicity of the irreducible representation $\Schur^{\lambda^\transpose} V$ in the decomposition of $\Schur^\nu \Wedge^m V$, where $\nu$ is the Young Diagram consisting of a single row of width $n$ if $m$ is even and $\nu$ equals a column of length~$n$ otherwise. If~$\lambda^\transpose$ is of width~$> n$, then we immediately deduce that $a_\lambda(n, m) = 0$. So assume that~$\lambda^\transpose$ is of width~$\leq n$ and construct a partition~$\pi^\transpose$ by prepending a row of width~$n$ to~$\lambda^\transpose$. From \cref{fact:reduction-wedge}, we infer that the multiplicity of \raisebox{0pt}[0pt][0pt]{$\Schur^{\pi^\transpose}$} in the decomposition of $\Schur^\nu \Wedge^{m+1}$ equals $a_\lambda(n, m)$ and by applying \cref{fact:duality} once more, we finally obtain that $a_\lambda(n, m) = b_\pi(n, m + 1)$.
\end{proof}

We are left to deal with the case $m = 3$.

% !TEX root = ../paper.tex

\section{Combinatorial Descriptions for Plethysm Coefficients via Discrete Tomography} \label{sec:plethysm-bounds}
Our results in \cref{sec:results} are based on a new combinatorial description of the plethysm coefficient in a subcase, based on new combinatorial lower and upper bounds for the plethysm coefficient. We think that the lower and upper bound as well as the combinatorial description are of independent interest.

By $[i, j]$, we denote the range $\{i, \ldots, j\} \subseteq \Nat$. We define $C := \{ (x, y, z) : x > y > z \} \subseteq \Nat^3$ as the \emph{open cone} and $\overline C = \{ (x, y, z) : x \geq y \geq z \} \subseteq \Nat^3$ as the \emph{closed cone}. A finite set $P \subseteq C$ is called a \emph{point set in the open cone}. A finite set $P \subseteq \overline{C}$ is called a \emph{point set in the closed cone}. A point set $P$ in the open cone is called an \emph{pyramid in the open cone} if for all $(x, y, z) \in P$ and all $(x', y', z') \in C$ with $x' \leq x$, $y' \leq y$, $z' \leq z$ we have $(x', y', z') \in P$.
Analogously, a point set $P$ in the closed cone is called a \emph{pyramid in the closed cone} if for all $(x, y, z) \in P$ and all $(x', y', z') \in \overline C$ with $x' \leq x$, $y' \leq y$, $z' \leq z$ we have $(x', y', z') \in P$.

Let $P$ be a point set in the open cone or closed cone. The \emph{sum-marginal} $S(P)$ is the sequence of numbers defined via
\begin{equation} \label{eq:symmetric-marginal}
  S_i(P) := \sum_{(x, y, z) \in P} \delta_{x, i} + \delta_{y, i} + \delta_{z, i},
\end{equation}
where $\delta_{i, j} = 1$ if $i = j$ and $\delta_{i, j}=0$ otherwise. Note that $|S(P)| := \sum_i S_i(P) = 3|P|$.

\begin{definition} \label{def:plethysm-bounds}
Let $\lambda$ be a partition of $3n$. We define:
\begin{itemize}
\item $\lowerbound a_\lambda(n, 3)$ as the number of pyramids in the open cone with sum-marginal $\lambda^\transpose$,
\item $\upperbound a_\lambda(n, 3)$ as the number of point sets in the open cone with sum-marginal $\lambda^\transpose$,
\item $\lowerbound b_\lambda(n, 3)$ as the number of pyramids in the closed cone with sum-marginal $\lambda$,
\item $\upperbound b_\lambda(n, 3)$ as the number of point sets in the closed cone with sum-marginal $\lambda$.
\end{itemize}
\end{definition}

The next theorem connects plethysm coefficients to the combinatorial notions that we just defined.

\begin{theorem} \label[theorem]{thm:plethysm-bounds}
Let $\lambda$ be a partition of $3n$. Then:
\begin{align*}
	\lowerbound a_\lambda(n, 3) \leq a_\lambda(n, 3) \leq \upperbound a_\lambda(n, 3) &&
	\text{and} &&
	\lowerbound b_\lambda(n, 3) \leq b_\lambda(n, 3) \leq \upperbound b_\lambda(n, 3).
\end{align*}
\end{theorem}
\begin{proof}
Let $X_0, X_1, \ldots, X_{k-1}$ be an ordered basis of $V = \Complex^k$. There is a natural bijection between the points in $\overline C_{< k} := [0, k-1]^3 \cap \overline C$ and the monomials in $\Sym^3 V$. It is easy to check that a monomial $X_x X_y X_z$ is of weight $\kappa$ if and only the corresponding point $(x, y, z)$ has sum-marginal~$\kappa$ (here, $\kappa$ is not necessarily a partition). In the same way, we match $(x, y, z) \in C_{< k} := [0, k-1]^3 \intersect C$ with elementary wedge products $X_x \wedge X_y \wedge X_z$. (Since we are dealing with these objects modulo scalars, the sign changes caused by the wedge product shall not bother us here.)

Now let $\lambda$ be a partition of $3n$. We associate each point set $P \subseteq \overline C_{< k}$ with sum-marginal~$\lambda$ with a vector
\[
  v_P := \bigwedge_{(x, y, z) \in P} X_x X_y X_z \in \Wedge^n \Sym^3 V;
\]
notice that $|\lambda| = 3n$ implies that $P$ is of size $n$ and thus indeed $v_P \in \bigwedge^n \Sym^3 V$. Analogously, for point sets $P \subseteq C_{< k}$ with sum-marginal $\lambda$, we define vectors
\[
	w_P := \bigwedge_{(x, y, z) \in P} X_x \wedge X_y \wedge X_z \in \Wedge^n \Wedge^3 V.
\]
From the preceding paragraphs, it follows that $v_P$ is a weight vector of weight $\lambda$, and furthermore, for $P$ ranging over all point subsets of $\overline C_{< k}$ with sum-marginal $\lambda$, $\{v_P\}_P$ forms a basis of the weight subspace of weight $\lambda$ of $\Wedge^n \Sym^3 V$. Recall that $\upperbound b_\lambda(n, 3)$ exactly counts the number of such sets $P$, thus $\upperbound b_\lambda(n, 3)$ equals the dimension of the weight-$\lambda$ subspace. Moreover, since $\Wedge^n \Sym^3 V$ is a $\GL(V)$-representation, the multiplicity $b_\lambda(n, 3)$ of $\lambda$ in the decomposition of $\Wedge^n \Sym^3 V$ into irreducible $\GL(V)$-representations matches the dimension of the linear subspace spanned by all highest weight vectors of weight $\lambda$. But each highest weight vector is in particular a weight vector and thus $b_\lambda(n, 3) \leq \upperbound b_\lambda(n, 3)$.

In the same way, $\{w_P\}_P$ forms a basis of the weight-$\lambda^\transpose$ subspace of $\Wedge^n \Wedge^3 V$, for $P$ ranging over all point subsets of $C_{< k}$ with sum-marginal $\lambda^\transpose$. Therefore, $\upperbound a_\lambda(n, 3)$ equals the dimension of the weight-$\lambda^\transpose$ subspace and hence upper-bounds the multiplicity of the irreducible component $\Schur^{\lambda^\transpose} V$ in $\Wedge^n \Wedge^3 V$. Recall that this multiplicity equals $a_\lambda(n, 3)$ by \cref{fact:duality}, so we conclude $a_\lambda(n, 3) \leq \upperbound a_\lambda(n, 3)$.

For the remainder of the proof, we can treat the symmetric (closed cone) and skew-symmetric (open cone) cases identically, so we omit the proof of $\lowerbound a_\lambda(n, 3) \leq a_\lambda(n, 3)$ and verify only that $\lowerbound b_\lambda(n, 3) \leq b_\lambda(n, 3)$. Let $P \subseteq \overline C_{< k}$ be a pyramid in the closed cone with sum-marginal $\lambda$, $|\lambda| = 3n$. It suffices to show that $v_P$ is not only a weight vector, but in fact a highest weight vector. To prove this, we demonstrate that $v_P$ is annihilated by all raising operators $E_{i, j}$; recall that, for $i < j$, $E_{i, j}$ denotes the upper triangular matrix with a single $1$ in the $i$-th row and the $j$-th column and zero entries anywhere else.
\begin{align*}
  E_{i, j} v_P &= \lim_{\varepsilon \to 0} \varepsilon^{-1} \big((\Id + \varepsilon E_{i, j}) v_P - v_P\big) \\
  &= \sum_{(x, y, z) \in P} (\delta_{x, j} X_i X_y X_z + \delta_{y, j} X_x X_i X_z + \delta_{z, j} X_x X_y X_i) \wedge \bigwedge_{\substack{(x', y', z') \in {}\\P \setminus \{(x, y, z)\}}} X_{x'} X_{y'} X_{z'}.
\end{align*}
We claim that each summand vanishes individually. Focus, for some $(x, y, z) \in P$, on the summand $\delta_{x, j} X_i X_y X_z \wedge \Wedge_{(x', y', z')} X_{x'} X_{y'} X_{z'}$, where $(x', y', z')$ is picked from $P \setminus \{(x, y, z)\}$. In case that $x \neq j$, this term clearly equals zero, so assume $x = j$. Since $i < j = x$, the points $(i, y, z)$ and $(x, y, z)$ differ. But $P$ is a pyramid and thus contains $(i, y, z)$, too. Consequently, the monomial $X_i X_y X_z$ also occurs in the big wedge product on the right-hand side and therefore cancels with $X_i X_y X_z$ on the left-hand side. The other summands with $i$ in place of $y$ or $z$, respectively, can be shown to vanish by an analogous argumentation.

Applying the raising operator $E_{i, j}$ to $v_P$ has an intuitive geometric interpretation for $P$: Any summand as above arises from $P$ by replacing some coordinate $j$ of some point in $P$ by $i$. Thereby that point is moved closer to the origin as $i < j$. Recall however, that pyramids are point sets not containing ``holes'', thus $P$ is transformed by mapping two points onto each other. Here comes the alternating property of the wedge product into play by annihilating the summand.
\end{proof}

When the lower and upper bound coincide, we get a combinatorial description of the (dual) plethysm coefficient. The following \cref{prop:min-coordinate-sum} gives a sufficient condition for when this happens. Let $\lambda \in \Nat^{[0, r]}$ be a composition. We define the \emph{coordinate sum $B(\lambda)$} via
\begin{align*}
	B(\lambda) = \sum_{i=0}^r i \cdot \lambda_i.
\end{align*}
Let $P$ be a point set in the open cone or closed cone. The coordinate sum $B(P)$ is defined as the coordinate sum of the sum-marginal $B(P) := B(S(P))$, which is the same as the sum over all coordinates:
\begin{align*}
	B(P) = \sum_{(x, y, z) \in P} x + y + z
\end{align*}
Let $\overline{\xi}(i)$ the number of cardinality-$1$ point sets in the closed cone with coordinate sum exactly $i$. This is the same as the number of partitions of $i$ of height $\leq 3$, which is OEIS sequence A001399, with explicit formula $\textsf{round}((i + 3)^2/12)$. Let $\xi(i)$ denote number of cardinality-$1$ point sets in the open cone with coordinate sum $i$. This is the same as the number of partitions $\mu$ of $i$ of height~$3$ such that $\mu_1$, $\mu_2$, and $\mu_3$ are pairwise distinct, which is OEIS sequence A069905\footnote{There is a bijection to partitions with 3 positive but not necessarily distinct parts via subtracting $(2, 1, 0)$.}, with explicit formula $\textsf{round}(i^2/12)$.

Now let $\overline\iota(n) := \min\{ \iota : \sum_{i=0}^\iota \overline\xi(i) \geq n\}$ and
analogously, let $\iota(n) := \min\{ \iota : \sum_{i=0}^\iota \xi(i) \geq n \}$. Finally, let $\overline\beta(n) := \sum_{i=1}^n \overline\iota(i)$ and analogously, let $\beta(n) := \sum_{i=1}^n \iota(i)$.

\begin{lemma}\label{lem:min-coordinate-sum}
There is no point set $P$ of cardinality $n$ in the closed cone that has $B(P) < \overline\beta(n)$. Analogously, there is no point set $P$ of cardinality $n$ in the open cone that has $B(P) < \beta(n)$.
\end{lemma}
\begin{proof}
This is seen by induction: If $P$ of cardinality $n$ in the open cone has $B(P) < \beta(n)$, then remove the point $(x, y, z)$ with highest coordinate sum to obtain the point set $P'$. Since for each~$\iota$ there are only $\xi(\iota)$ many points $(x', y', z')$ in $C$ with $x' + y' + z' = \iota$, it follows that $x + y + z \geq \iota(n)$. But $\beta(n) - \iota(n) = \beta(n - 1)$, which implies $B(P') \leq B(P) - \iota(n) < \beta(n) - \iota(n) = \beta(n - 1)$, which means that $P'$ cannot exist by induction hypothesis and hence $P$ cannot exist. A completely analogous argument works for the closed cone.
\end{proof}

\begin{proposition} \label{prop:min-coordinate-sum}
If $B(\lambda) = \overline\beta(|\lambda| / 3)$, then each point set in the closed cone with sum-marginal $\lambda$ is a pyramid in the closed cone. In particular, if $\lambda$ is a partition of $3n$, according to \cref{thm:plethysm-bounds}, we have $\lowerbound b_\lambda(n, 3) = b_\lambda(n, 3) = \upperbound b_\lambda(n, 3)$.

If $B(\lambda^\transpose) = \beta(|\lambda| / 3)$, then each point set in the open cone with sum-marginal $\lambda^\transpose$ is a pyramid in the open cone. In particular, if $\lambda$ is a partition of $3n$, then according to \cref{thm:plethysm-bounds}, we have $\lowerbound a_\lambda(n, 3) = a_\lambda(n, 3) = \upperbound a_\lambda(n, 3)$.
\end{proposition}
\begin{proof}
Assume for the sake of contradiction that there exists a point set $P$ in the closed cone with sum-marginal $S(P) = \lambda$ (note that $S(P) = \lambda$ implies $|P| = |\lambda| / 3$) such that $P$ is not a pyramid in the closed cone. Then there exists $(x, y, z) \in P$ and $(x', y', z') \notin P$ with $x' \leq x$, $y' \leq y$, $z' \leq z$. Replacing $(x, y, z)$ in $P$ with $(x', y', z')$ gives a point set $P'$ with $B(P') < B(P)$. This is a contradiction to \cref{lem:min-coordinate-sum}. A completely analogous argument works for the open cone.
\end{proof}

% !TEX root = ../paper.tex

\section{Proof of Main Theorem~\ref{thm:plethysm-sharpp-hard}} \label{sec:reduction}

The reductions that we present are parsimonious polynomial-time reductions between counting problems. Such reductions automatically yield polynomial-time many-one reductions between the associated decision problems.

\subsection*{Counting and Decision Problems in Discrete Tomography}
Let us start with some notations: Let $P \subseteq \Nat^3$ be a finite set of points. By $X_i(P)$ we address the number of elements in $P$ with $X$-coordinate $i$. Similarly, $Y_i(P)$ and $Z_i(P)$ count the elements with $Y$- and $Z$-coordinate $i$, respectively. We call $(X_i(P))_i$, $(Y_i(P))_i$ and $(Z_i(P))_i$ the \emph{$X$-, $Y$-} and \emph{$Z$-marginals} of $P$; these sequences are treated as finite by omitting trailing zeros. Recall the definition \eqref{eq:symmetric-marginal} of the sum-marginal $S(P)$ and observe that $S(P) = X(P)+Y(P)+Z(P)$.

The intermediate 2-dimensional problems below operate on the grid
\[
  G_r := \{ (x, y, z) : x + y + z = r \} \subseteq \Nat^3;
\]
in spirit, $G_r$ is two-dimensional although each entry of~$G_r$ is parameterized by~$3$ coordinates (the so-called trilinear coordinates).
This specific embedding of $G_r \subseteq \Nat^3$ plays a major role in the proof of \cref{lem:reduction-2dxray-3dxray}.

% !TEX root = ../paper.tex

\begin{figure}[t]
\begin{subfigure}[b]{0.48\textwidth}
\centering
\begin{tikzpicture}[triangle diagram, r=7, triangular=.45cm]

\drawtriaxesnoarrows
\drawtriaxesarrowsxyz
\drawtriaxescapt
\drawtrigrid

\drawpoints{
  0/3/4, 0/6/1,
  1/4/2, 1/5/1,
  2/1/4,
  4/0/3, 4/2/1, 4/3/0,
  6/1/0
}

\end{tikzpicture}
\end{subfigure}
\hfill
\begin{subfigure}[b]{0.48\textwidth}
\centering
\begin{tikzpicture}[triangle diagram, r=6, triangular=.45cm]

\def\axisovershoot{.7cm}

\fillcone
\drawtriaxesnoarrows
\drawtriaxesarrowsxyz
\drawtriaxescapt
\drawtrigrid

\drawpoints{
  2/2/2,
  3/2/1,
  5/1/0
}

\end{tikzpicture}
\end{subfigure}
\caption[paragraph]{The left depiction of $P \subseteq G_7$ certifies that the \TwoDXRay{} instance defined by $\mu = (2, 2, 1, 0, 3, 0, 1, 0)$, $\nu = (1, 2, 1, 2, 1, 1, 1, 0)$ and $\rho = (2, 3, 1, 1, 2, 0, 0, 0)$ is satisfiable.

The right image gives an example of a set $P \subseteq G_6 \intersect \overline C$ with sum-marginal $\lambda = (0, 0, 1, 1, 0, 1, 0) + (0, 1, 2, 0, 0, 0, 0) + (1, 1, 1, 0, 0, 0, 0) = (1, 2, 4, 1, 0, 1, 0)$. The framed area (including the boundary) illustrates the closed cone~$\overline C$, while the interior of that area illustrates the open cone~$C$.} \label{fig:xray-example}
\end{figure}
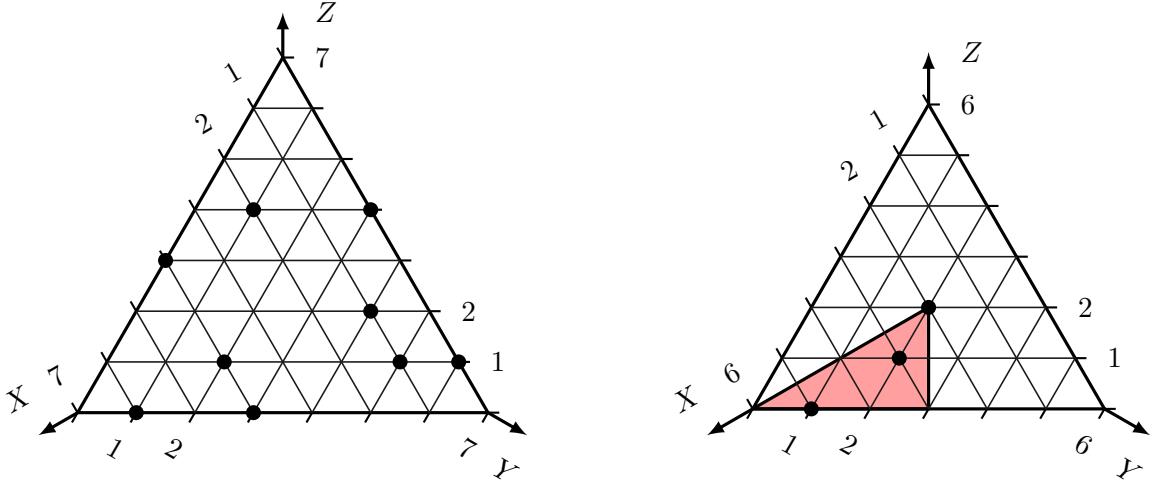

\begin{problem}[\TwoDXRay{}]
Given $\mu, \nu, \rho \in \Nat^{[0, r]}$, the \emph{\TwoDXRay{}} problem is to decide if there exists a point set $P \subseteq G_r$ with $X$-, $Y$- and $Z$-marginals $\mu$, $\nu$ and $\rho$, respectively.
\end{problem}

\cref{fig:xray-example} exemplifies an instance of \TwoDXRay{} with $r = 7$. For this problem and all following tomography problems, we denote the respective counting versions by prepending \Sharp{}. For example:

\begin{problem}[\Sharp\TwoDXRay{}]
Given $\mu, \nu, \rho \in \Nat^{[0, r]}$, the \emph{\TwoDXRay{}} problem is to compute the number of point sets $P \subseteq G_r$ with $X$-, $Y$- and $Z$-marginals $\mu$, $\nu$ and $\rho$, respectively.
\end{problem}

\TwoDXRay{} has been identified as \NP{}-hard (even \NP{}-complete) before by Gardner, Gritzmann and Prangenberg~\cite{GardnerGP1999}. In fact, they prove that the counting version \Counting\TwoDXRay{} is \SharpP{}-hard (even \SharpP{}-complete).
Keeping this in mind, we now continue to reduce \Sharp\TwoDXRay{} in a parsimonious way to computing $a_\lambda(n, 3)$ or $b_\lambda(n, 3)$, taking the detour over the other \Sharp\XRay{} problems. Our reductions are illustrated in \cref{fig:reductions}. % !TEX root = ../paper.tex

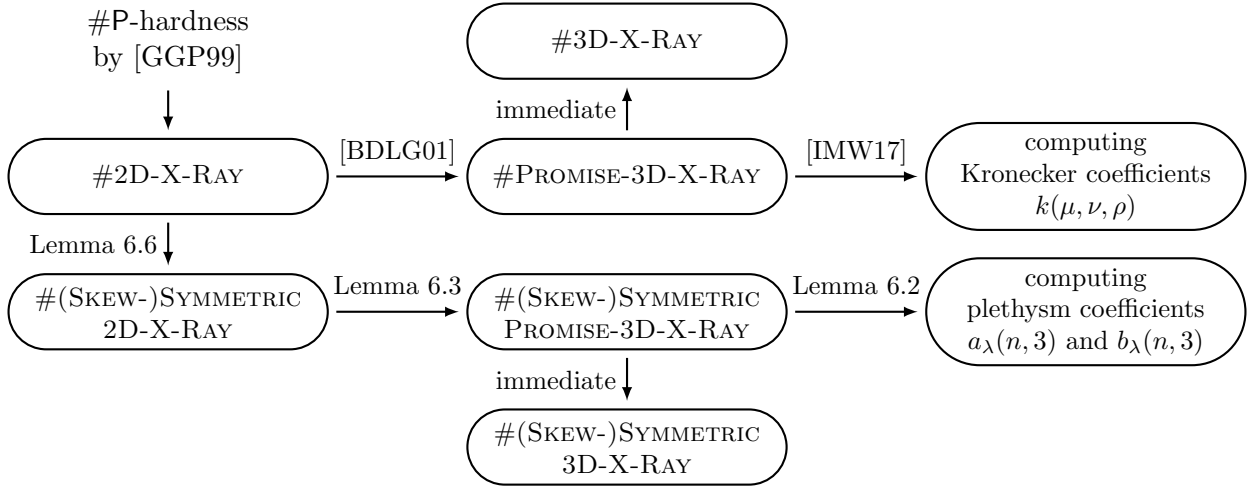
\begin{figure}
\begin{center}
\begin{tikzpicture}[
	problem/.style={
		font=\small,
		align=center,
		draw,
		thick,
		rounded rectangle,
		minimum width=4.5cm,
		minimum height=1cm
	},
	reduction/.style={
		font=\small,
		align=center,
		->,
		>=latex,
		thick,
		shorten >=.1cm,
		shorten <=.1cm
	}
]
\def\x{6.1cm}
\def\y{-1.8cm}
\path
	(0, 0) node[align=center] (initial) {\SharpP{}-hardness\\by~\cite{GardnerGP1999}}
	++(0, \y) node[problem] (TwoDXRay) {\#\TwoDXRay{}}
	+(0, \y) node[problem] (symmetric TwoDXRay) {\#\textsc{(Skew-)Symmetric}\\\TwoDXRay{}}
	++(\x, 0) node[problem] (ThreeDXRay) {\#\PromiseThreeDXRay{}}
	+(0, \y) node[problem] (symmetric ThreeDXRay) {\#\textsc{(Skew-)Symmetric}\\\PromiseThreeDXRay{}}
	+(0, -\y) node[problem] (nopromise ThreeDXRay) {\#\ThreeDXRay{}}
	+(0, \y+\y) node[problem] (nopromise symmetric ThreeDXRay) {\#\textsc{(Skew-)Symmetric}\\\ThreeDXRay{}}
	++(\x, 0) node[problem] (Kronecker) {computing\\Kronecker coefficients\\$k(\mu, \nu, \rho)$}
	+(0, \y) node[problem] (plethysm) {computing\\plethysm coefficients\\$a_\lambda(n, 3)$ and $b_\lambda(n, 3)$};

\draw[reduction] (initial) to (TwoDXRay);
\draw[reduction] (ThreeDXRay) to node[left] {immediate} (nopromise ThreeDXRay);
\draw[reduction] (symmetric ThreeDXRay) to node[left] {immediate} (nopromise symmetric ThreeDXRay);
\draw[reduction] (TwoDXRay) to node[left] {\cref{lem:reduction-2dxray-sym2dxray}} (symmetric TwoDXRay);
\draw[reduction] (TwoDXRay) to node[above] {\cite{BrunettiDG2001}} (ThreeDXRay);
\draw[reduction] (ThreeDXRay) to node[above] {\cite{IkenmeyerMW2017}} (Kronecker);
\draw[reduction] (symmetric TwoDXRay) to node[above] {\cref{lem:reduction-2dxray-3dxray}\vphantom[} (symmetric ThreeDXRay);
\draw[reduction] (symmetric ThreeDXRay) to node[above] {\cref{lem:reduction-3dxray-plethysm}\vphantom[} (plethysm);
\end{tikzpicture}
\end{center}
\caption{The reductions that lead to \SharpP{}-hardness proofs for the problems of computing Kronecker and plethysm coefficients. All depicted reductions are parsimonious and polynomial-time, so the the \NP{}-hardness of all corresponding decision problems ultimately follows from the \NP{}-hardness of \TwoDXRay{} that was proved in \cite{GardnerGP1999}.} \label{fig:reductions}
\end{figure}

\begin{problem}[\Symmetric\TwoDXRay{} and \SkewSymmetric\TwoDXRay{}]
Given $\lambda \in \Nat^{[0, r]}$, the \emph{\Symmetric\TwoDXRay{}} problem is to decide if there exists a point set $P \subseteq G_r \intersect \overline{C}$ with sum-marginal~$S(P) = \lambda$.

Analogously, given $\lambda \in \Nat^{[0, r]}$, the \emph{\SkewSymmetric\TwoDXRay{}} problem is to decide if there exists a point set $P \subseteq G_r \intersect C$.
\end{problem}

Finally, we introduce three-dimensional variants that differ from the above problems only in that $P$ is taken from $\Nat^3$ rather than $G_r$:

\begin{problem}[\ThreeDXRay{}]
Given $\mu, \nu, \rho \in \Nat^{[0, r]}$, the \emph{\ThreeDXRay{}} problem is to decide if there exists a point set $P \subseteq \Nat^3$ with $X$-, $Y$- and $Z$-marginals $\mu$, $\nu$ and $\rho$, respectively.
\end{problem}

\ThreeDXRay{} is known to be \NP-complete by \cite{BrunettiDG2001}, see Figure~\ref{fig:reductions}.

\begin{problem}[\Symmetric\ThreeDXRay{} and \SkewSymmetric\ThreeDXRay{}]
Given $\lambda \in \Nat^{[0, r]}$, the \emph{\Symmetric\ThreeDXRay{}} problem is to decide if there exists a point set $P \subseteq \overline C$ with sum-marginal~$\lambda$.

Analogously, given $\lambda \in \Nat^{[0, r]}$, the \emph{\SkewSymmetric\ThreeDXRay{}} problem is to decide if there exists a point set $P \subseteq C$ with sum-marginal~$\lambda$.
\end{problem}

Our main focus in 3D will be the following \emph{promise problem version}.

\begin{problem}[\Symmetric\PromiseThreeDXRay{} and \SkewSymmetric\PromiseThreeDXRay{}]\hfill
Given $\lambda \in \Nat^{[0, r]}$ such that $B(\lambda) = \overline\beta(|\lambda| / 3)$, the \emph{\Symmetric\PromiseThreeDXRay{}} problem is to decide if there exists a point set $P \subseteq \overline C$ with sum-marginal~$\lambda$.

Analogously, given $\lambda \in \Nat^{[0, r]}$ such that $B(\lambda) = \beta(|\lambda| / 3)$, the \emph{\SkewSymmetric\PromiseThreeDXRay{}} problem is to decide if there exists a point set $P \subseteq C$ with sum-marginal~$\lambda$.
\end{problem}

\subsection*{Reduction from \Sharp{\normalfont\textsc{(Skew-)Symmetric}} \PromiseThreeDXRay{} to {\normalfont\textsc{(Dual-)}}\Plethysm}

Recall that {\normalfont\textsc{(Skew-)Symmetric}} \PromiseThreeDXRay{} instances~$\lambda$ need not be partitions in general. Still, we prove that any promise instance~$\lambda$ which is not a partition is trivially rejected:

\begin{lemma} \label[lemma]{lem:pyramid-partition}
Let~$\lambda$ be a composition, and assume there exists some pyramid $P$ in the open or closed cone with sum-marginal $\lambda$. Then $\lambda$ is a partition.
\end{lemma}
\begin{proof}
The proof is slightly subtle. Fix an arbitrary $i$; we want to show that $\lambda_i \geq \lambda_{i+1}$.
Let us define $\alpha(x, y, z) := \delta_{x, i} + \delta_{y, i} + \delta_{z, i} - \delta_{x, i+1} - \delta_{y, i+1} - \delta_{z, i+1}$. It is easy to see that $\lambda_i - \lambda_{i+1} = \sum_{(x, y, z) \in P} \alpha(x, y, z)$, so we have to show that $\sum_{(x, y, z) \in P} \alpha(x, y, z) \geq 0$. We achieve this as follows: To each point $(x, y, z) \in P$ with $\alpha(x, y, z) < 0$, we assign a point $(x', y', z') \in P$ such that $\alpha(x, y, z) + \alpha(x', y', z') = 0$. Every $(x', y', z')$ will be used for at most one $(x, y, z)$, which then finishes the proof. We proceed to define $(x', y', z') =: \varphi(x, y, z)$.
The symbol $\cdot$ matches any input $\notin \{i, i-1\}$ and is not modified in the output. We set
$\varphi(i+1, i+1, i+1) := (i, i, i)$,
$\varphi(i+1, i+1, i) := (i+1, i, i)$,
$\varphi(i+1, i+1, \cdot) := (i, i, \cdot)$,
$\varphi(\cdot, i+1, i+1) := (\cdot, i, i)$,
$\varphi(i+1, \cdot, \cdot) := (i, \cdot, \cdot)$,
$\varphi(\cdot, i+1, \cdot) := (\cdot, i, \cdot)$,
$\varphi(\cdot, \cdot, i+1):=(\cdot, \cdot,i)$.
The map $\varphi$ maps points in $C$ to points in $C$, and points in $\overline C$ to points in $\overline C$.
Moreover, since $P$ is a pyramid, we have indeed $(x',y',z') \in P$.
\end{proof}

\begin{lemma} \label[lemma]{lem:reduction-3dxray-plethysm}
There exists a parsimonious polynomial-time reduction from \Sharp\Symmetric\PromiseThreeDXRay{} to
\DualPlethysm\Inner{$3$}.

Moreover, there exists a parsimonious polynomial-time reduction from \Sharp\SkewSymmetric\PromiseThreeDXRay{} to
\Plethysm\Inner{$3$}.
\end{lemma}
\begin{proof}
We use a trivial no-instance for \DualPlethysm\Inner{$3$}: $b_{(2, 1)}(1, 3) = 0$. Given an instance $(\lambda, n)$ of \Sharp\Symmetric\PromiseThreeDXRay{}, if $\lambda$ is not a partition or if $|\lambda| \neq 3n$, then the reduction outputs that trivial no-instance. This is the correct behavior according to \cref{lem:pyramid-partition}.

On the other hand, if $\lambda$ is a partition, then the reduction outputs its input $(\lambda, n)$. Since $B(\lambda) = \overline\beta(|\lambda| / 3)$, \cref{prop:min-coordinate-sum} ensures that $b^+_\lambda(n, 3) = b_\lambda(n, 3)$. Observe that in the language of \cref{sec:plethysm-bounds}, \Sharp\Symmetric\PromiseThreeDXRay{} asks to compute $b^+_\lambda(n, 3)$. Hence the reduction works correctly.

The same argument works for $a_\lambda(n, 3)$ and $\beta(|\lambda| / 3)$. Here, in the nontrivial case, $(\lambda^\transpose, n)$ is returned.
\end{proof}

\subsection*{Reduction from \Sharp\Symmetric\TwoDXRay{} to \Sharp\SkewSymmetric\PromiseThreeDXRay{}}
In this section we prove the following lemma.

\begin{lemma} \label[lemma]{lem:reduction-2dxray-3dxray}
There exists a parsimonious polynomial-time reduction from \Sharp\Symmetric\TwoDXRay{} to \Sharp\Symmetric\PromiseThreeDXRay{}.

Moreover, there exists a parsimonious polynomial-time reduction from \Sharp\SkewSymmetric\TwoDXRay{} to \Sharp\SkewSymmetric\PromiseThreeDXRay{}.
\end{lemma}

Before we can prove \cref{lem:reduction-2dxray-3dxray} we make an observation (\cref{lem:coordinate-sum-pyramid} below) of independent interest.
We call $P_r := \{ (x, y, z) : x + y + z \leq r, x > y > z \}$ the \emph{complete pyramid in the open cone},
and analogously $\overline{P_r} := \{ (x, y, z) : x + y + z \leq r, x \geq y \geq z \}$ the \emph{complete pyramid in the closed cone}.
Note that for all $r$ we have $\overline\beta(|\overline P_r|) = B(\overline P_r)$. Moreover, $\overline\beta(n)$ is the piecewise linear interpolation between these values.
Analogously, $\beta(n)$ is the piecewise linear interpolation between $\beta(|P_r|) = B(P_r)$.

The proof of the following lemma is a careful adaption of a proof given by Brunetti, Del Lungo and Gérard~\cite{BrunettiDG2001} for the non-symmetric \ThreeDXRay{} problem.
Roughly speaking, the trick is to embed the two-dimensional grid $G_r$ into the $r$-th layer of the complete pyramid.

\begin{lemma} \label[lemma]{lem:coordinate-sum-pyramid}
Let $P \subseteq \overline C$ be of size $|P| = n$. Then $P$ has coordinate sum $B(P) = \overline\beta(n)$ if and only if $\overline P_{r-1} \subseteq P \subseteq \overline P_r$ for some $r$.

Analogously, let $P \subseteq C$ be of size $|P| = n$. Then $P$ has coordinate sum $B(P) = \beta(n)$ if and only if $P_{r-1} \subseteq P \subseteq P_r$ for some $r$.
\end{lemma}
\begin{proof}
It is easy to see by a direct calculation that all point sets $\overline P_{r-1} \subseteq P \subseteq \overline P_r$ have $B(P)=\overline\beta(n)$, since every point in $P\setminus \overline P_{r-1}$ has the same coordinate sum~$r$ (i.e., $P\setminus \overline P_{r-1} \subseteq G_r$).

The other direction is intuitively straightforward by the following argument: Whenever $P$ contains a ``hole'', then we can fill the hole by moving outer points closer to the origin. But in that we decrease the coordinate sum of $P$. More formally: Assume that $P$ has minimal coordinate sum among all sets $\subseteq \overline C$ of the same size $n$. We pick $r$ to be maximal, such that $\overline P_{r-1} \subseteq P$ and assume, for the sake of contradiction, that $P \not\subseteq \overline P_r$, i.e., there exists some $(x, y, z) \in P \setminus \overline P_r$. Furthermore, there exists some $(x', y', z') \in \overline P_r \setminus P$, since otherwise $P \subseteq \overline P_r$ which contradicts the way we chose~$r$. We construct the set $P' := P \union \{ (x', y', z') \} \setminus \{ (x, y, z) \}$; clearly $|P'| = |P| = n$. The coordinate sum of $P'$ is bounded by
\[
  B(P') = B(P) - \underbrace{(x + y + z)}_{{} > r} + \underbrace{(x' + y' + z')}_{{} \leq r} < B(P) = \overline\beta(n),
\]
which is a contradiction to the definition of $\overline\beta(n)$.

This proof works completely analogously for $C$.
\end{proof}

The reduction now proceeds by embedding $G_r$ in the three-dimensional space. Any solution $\widehat P$ of the given {\normalfont\textsc{(Skew-)Symmetric}} \TwoDXRay{} instance is interpreted as the $r$-th layer of the thereunder entirely filled pyramid $P = \overline P_{r-1} \union \widehat P$; \cref{fig:pyramid} might give some further intuition on how to view {\normalfont\textsc{Symmetric}} \TwoDXRay{} in three-dimensional space. We detail the construction in the following lemma:

% !TEX root = ../paper.tex

\begin{figure}[t]
\begin{subfigure}[b]{0.44\textwidth}
\centering
\vspace{-.5cm}
\begin{tikzpicture}[triangle diagram, r=4, three dimensional skew=1.45cm]

\def\axiscapt{.18cm}
\def\labelr{$r$}

\drawthreedaxes
\drawthreedaxesskewxyz

\foreach \i in {1, ..., 4} {
  \edef\r{\i}
  \drawtriframe
	\fillcone
	\drawtrigrid
}

\drawthreedskewxlabel{1}{$G_1$}
\drawthreedskewxlabel{2}{$G_2$}
\drawthreedskewxlabel{\r}{$G_r$}

\end{tikzpicture}
\end{subfigure}
\hfill
\begin{subfigure}[b]{0.55\textwidth}
\centering
\begin{tikzpicture}[triangle diagram, r=5, three dimensional straight=.65cm]

\def\axiscapt{.18cm}
\def\labelr{$r$}

\drawthreedaxes
\drawtriaxesarrowsxyz
\drawtriframe
\drawtriaxescapt
\drawthreedgrid
\fillcone
\drawtrigrid

\end{tikzpicture}
\end{subfigure}
\caption{Visualizes how to arrange the grids $G_r$ in three-dimensional space; the red-colored areas indicate the intersections with $\overline C$ (or $C$, by restricting to the interior). The essential step of the reduction is to embed a set $\widehat P \subseteq G_r$ into the top-most layer of the pyramid $P := P_{r-1} \union \widehat P$.} \label{fig:pyramid}
\end{figure}
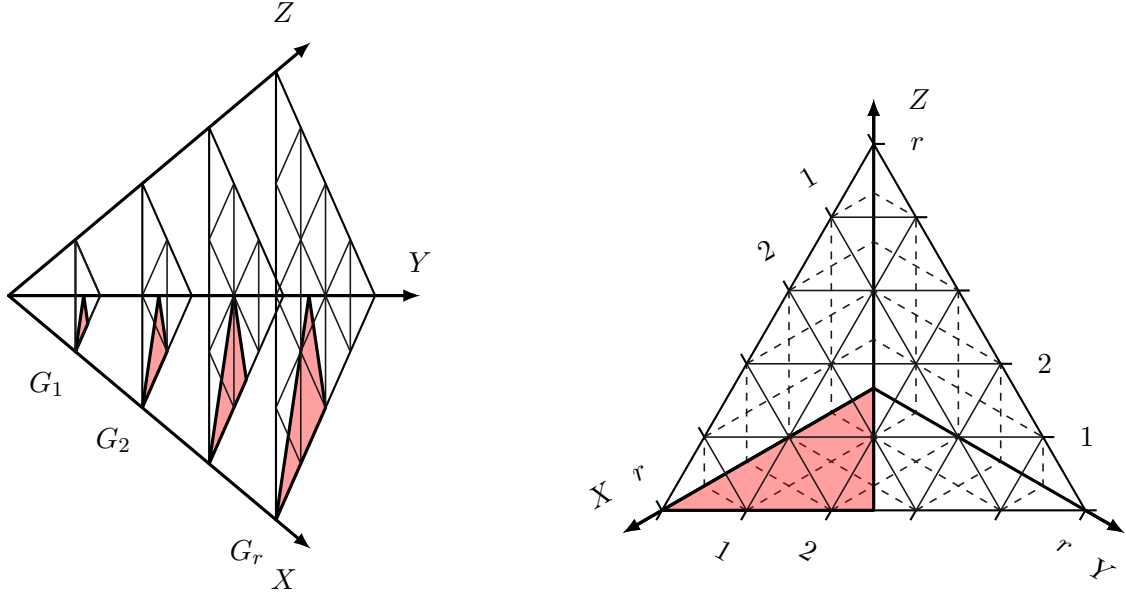

\begin{lemma} \label[lemma]{lem:pyramid-embedding}
Fix some $r$. Let $\widehat\lambda \in \Nat^{[0, r]}$ be such that there exists $n$ with $|\widehat\lambda| = 3n$ and $B(\widehat\lambda) = rn$. (i.e., $\widehat\lambda$ is not trivially rejected as a {\normalfont\textsc{Symmetric}} \TwoDXRay{} instance). Let $\lambda := S(\overline P_{r-1}) + \widehat\lambda$. Then $B(\lambda) = \overline\beta(|\overline P_{r-1}| + n)$ and there is a one-to-one correspondence between point sets $P \subseteq \overline C$ with sum-marginal $\lambda$ and point sets $\widehat P \subseteq G_r \intersect \overline C$ with sum-marginal $\widehat\lambda$.

The same statements hold if instead we consider point sets in the open cone $C$, $P_{r-1}$ and $P_r$, and the {\normalfont\textsc{Skew-Symmetric}} problems. In that case $\lambda := S(P_{r-1}) + \widehat\lambda$.
\end{lemma}
\begin{proof}
We calculate $B(\lambda) = B(\overline P_{r-1}) + rn = \overline\beta(|\overline P_{r-1}|) + rn = \overline\beta(|\overline P_{r-1}| + n)$. The claimed one-to-one correspondence works as follows: For $P \subseteq \overline C$ with sum-marginal $\lambda = S(\overline P_{r-1}) + \widehat\lambda$ we know (\cref{lem:coordinate-sum-pyramid}) that $P = \overline P_{r-1} \union \widehat P$ for some $\widehat P \subseteq \overline C \intersect G_r$. The set $\widehat P$ has sum-marginal $\widehat\lambda$. Conversely, if some set $\widehat P$ has sum-marginal $\widehat\lambda$, then the set $P := \overline P_{r-1} \union \widehat P$ has sum-marginal $S(\overline P_{r-1}) + \widehat\lambda = \lambda$. Both maps $P \mapsto \widehat P$ and $\widehat P \mapsto P$ are inverse to each other, which finishes the proof.

After applying the obvious changes, the proof also works for the open cone $C$.
\end{proof}

\begin{proof}[Proof of \cref{lem:reduction-2dxray-3dxray}]
Given a \Sharp\Symmetric\TwoDXRay{} instance $\widehat\lambda$, we first assert that $|\widehat\lambda|$ is divisible by~$3$ and that $B(\widehat\lambda_i) = rn$ holds for $n = |\widehat\lambda|/3$. If this step fails, output a trivially zero \Sharp{\normalfont\textsc{Symmetric}} \ThreeDXRay{} instance. Otherwise, construct $\lambda := S(\overline P_{r-1}) + \widehat\lambda$ and output $\lambda$ as the corresponding \Sharp{\normalfont\textsc{Symmetric}} \ThreeDXRay{} instance. Clearly, this algorithm runs in polynomial time.

\cref{lem:pyramid-embedding} entirely proves the remaining goals: Since $B(\lambda) = \overline\beta(|\overline P_{r-1}| + n)$, the promise is kept. Moreover, the one-to-one correspondence in \cref{lem:pyramid-embedding} ensures that the reduction is correct and parsimonious. Again, the same proof works if one replaces \Symmetric{}by \SkewSymmetric{}and~$\overline C$ by~$C$.
\end{proof}

\subsection*{Reduction from \Sharp\TwoDXRay{} to \Sharp\Symmetric\TwoDXRay{}}
The following \cref{lem:reduction-2dxray-sym2dxray} is the first known result that connects \XRay{} problems with their symmetric counterparts.

\begin{lemma} \label[lemma]{lem:reduction-2dxray-sym2dxray}
There exist parsimonious polynomial-time reductions from \Sharp\TwoDXRay{} to \Sharp\Symmetric\TwoDXRay{} and to \Sharp\SkewSymmetric\TwoDXRay{}.
\end{lemma}
\begin{proof}
Let $\mu', \nu', \rho' \in \Nat^{[0, r']}$ be a given \Sharp\TwoDXRay{} instance. We start to construct the corresponding \Sharp{\normalfont\textsc{Symmetric}} \TwoDXRay{} instance by choosing $r := 13 r'$ and
\begin{align*}
  \lambda := (\rho'_0, \ldots, \rho'_{r'}, \underbrace{0, \ldots, 0}_{2r' - 1}, \nu'_0, \ldots, \nu'_{r'}, \underbrace{0, \ldots, 0}_{5r' - 1}, \mu'_0, \ldots, \mu'_{r'}, \underbrace{0, \ldots, 0}_{3r'}) \in \Nat^{[0, r]}.
\end{align*}

% !TEX root = ../paper.tex

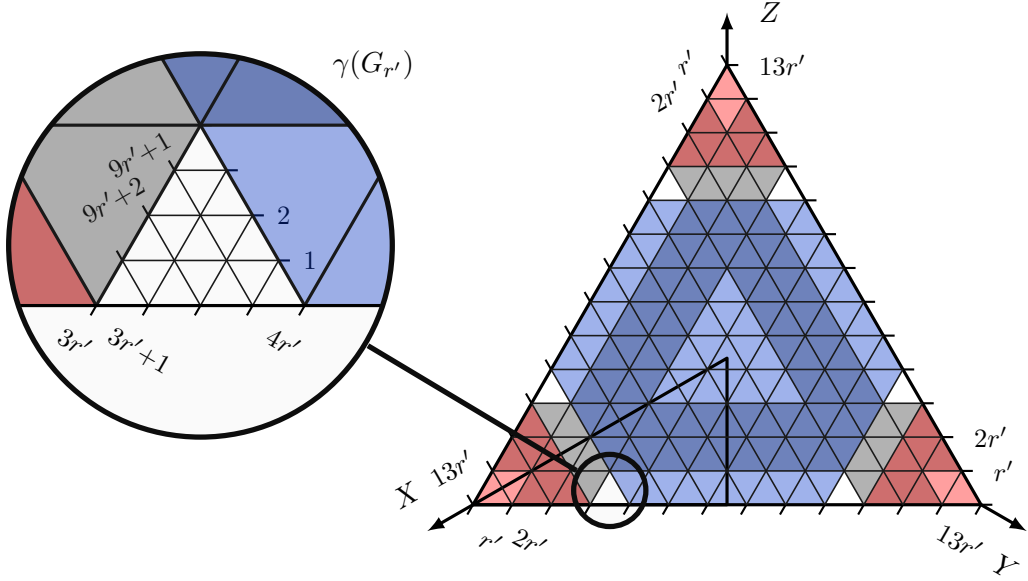
\begin{figure}[t]
\begin{center}
\begin{tikzpicture}[triangle diagram, r=13, triangular=.3cm]

\def\axisovershoot{.7cm}
\def\labelone{\small $r'$}
\def\labeltwo{\small $2r'$}
\def\labelr{\small $13r'$}

\def\black{black!80}
\def\green{green!20!blue}
\def\red{red}

\fillarea{1}{3}{\black}
\fillarea{4}{9}{\green}
\fillarea{10}{13}{\red}

\drawcone
\drawtriaxesnoarrows
\drawtriaxesarrowsxyz
\drawtriaxescapt
\drawtrigrid

\begin{scope}[triangle diagram, r=4, triangular=.4cm, xshift=-7cm, yshift=1.5cm]
	\clip (0, 0) circle (2.55cm);

	\def\labelzero{}
	\def\labelone{\small\hspace{-.6cm}\vbox to .1cm{$9r' \mathord+ 1$}}
	\def\labeltwo{\small\hspace{-.6cm}\vbox to .1cm{$9r' \mathord+ 2$}}
	\def\labelr{}
	\drawtriaxescaptx

	\def\labelzero{\small$3r'$}
	\def\labelone{\small\hspace{.4cm}$3r' \mathord+ 1$}
	\def\labeltwo{}
	\def\labelr{\small $4r'$}
	\drawtriaxescapty

	\def\labelzero{}
	\def\labelone{\small\hspace{-.1cm}$1$}
	\def\labeltwo{\small\hspace{-.1cm}$2$}
	\def\labelr{}
	\drawtriaxescaptz

	\drawtrigrid

	\begin{scope}[triangle diagram, r=3, triangular=1.6cm, yshift=1.6cm, grid/.append style={axis}]
		\filltriup{3}{0}{0}{\black}
		\filltriup{3}{0}{0}{\red}
		\filltridown{2}{0}{1}{\black}
		\filltriup{2}{0}{1}{\black}
		\filltridown{1}{0}{2}{\black}
		\filltridown{1}{0}{2}{\green}
		\filltriup{1}{1}{1}{\black}
		\filltriup{1}{1}{1}{\green}
		\filltridown{1}{1}{1}{\green}
		\filltriup{1}{2}{0}{\green}

		\drawtrigrid
		\drawtriaxes
	\end{scope}
\end{scope}

\begin{scope}[fill=black!95, fill opacity=.02]
	\coordinate (zoom in) at (-7cm, 1.5cm);
	\coordinate (zoom out) at ($(\r, 0, 0)!3.5/13!(0, \r, 0) + (0, .19cm)$);
	\node[draw=black!95, line width=.07cm, fill, circle, minimum size=5.1cm] (zoom in circle) at (zoom in) {};
	\node[draw=black!95, line width=.07cm, fill, circle, minimum size=.96cm] (zoom out circle) at (zoom out) {};
	\path[draw=black!95, line width=.07cm] (zoom in circle) -- (zoom out circle);
\end{scope}
\node at (-4.7cm, 3.9cm) {$\gamma(G_{r'})$};

\end{tikzpicture}
\end{center}
\vspace{-.4cm}
\caption[paragraph]{
Visualizes the reduction from \TwoDXRay{} to \Symmetric\TwoDXRay{}. The lattice depicts the subdivision of the triangular grid~$G_r$ into~$13$ regions per coordinate, each of width~$r'$. The intervals~$(r', 3r')$, $(4r', 9r')$ and~$(10r', 13r']$---that is, the blocks of zeros in~$\lambda$---are respectively indicated by gray-, blue- and red-colored areas; for instance, a point is colored gray if at least one of its coordinates lies in~$(r', 3r')$.

Observe that the circled region equals the image of $\gamma$ under $G_{r'}$, which is the only uncolored region intersecting the (closed) cone $\overline C$.} \label{fig:symmetrization}
\end{figure}

We claim that the \Sharp\Symmetric\TwoDXRay{} instance $\lambda$ has the same number of solutions as the \Sharp\TwoDXRay{} instance $(\mu', \nu', \rho')$. Indeed, consider the following map $\gamma$ that maps points in $G_{r'}$ to points in $G_{r} \intersect \overline C$:
\begin{align*}
	\gamma(x, y, z) = (x + 9r', y + 3r', z)
\end{align*}
We lift $\gamma$ from points to point sets in the canonical way and abuse notation by calling the resulting map $\gamma$ again. We claim that $\gamma$ maps bijectively
\begin{align*}
	\gamma :\quad
	\left\{\text{\parbox{5cm}{\centering point sets in $G_{r'}$ with\\$(X, Y, Z)$-marginal $(\mu', \nu', \rho')$}}\right\}
	\quad\longrightarrow\quad
	\left\{\text{\parbox{3.8cm}{\centering point sets in $G_r \intersect \overline C$\\ with sum-marginal $\lambda$}}\right\},
\end{align*}
which then finishes the proof. An illustration of $\gamma$ is provided in \cref{fig:symmetrization}. The well-definedness and the injectivity of $\gamma$ are obvious. It remains to prove the surjectivity of $\gamma$. We prove surjectivity by constructing an explicit preimage. To achieve this, we start with some observations on the codomain of~$\gamma$. So let $P$ be a point set in $G_r \intersect \overline C$ with sum-marginal $\lambda$. Let $M := [8r', 9r']$, $N := [3r', 4r']$ and $R := [0, r']$. We first establish that for all points $(x, y, z) \in P$ it holds that $x$, $y$ and $z$ are all contained in $M$, $N$ or $R$ and neither two are included in the same interval. As a first insight, since $\lambda_i = 0$ for all $i \not\in M \union N \union R$, it follows that $x$, $y$ and $z$ are indeed contained in $M \union N \union R$. In order to prove that the coordinates $x$, $y$ and $z$ must stem from pairwise distinct sets $M$, $N$ and $R$, we distinguish the following scenarios and show that each case causes a contradiction:

\medskip
\emph{Two coordinates in $R$:} If, say, $y, z \in R$, then $x = r - y - z \geq 13r' - r' - r' = 11r'$. But the range $[11r', 13r']$ intersects none of $M$, $N$ and $R$.

\medskip
\emph{Two coordinates in $N$:} If, say, $x, y \in N$, then $z = r - x - y \geq 13r' - 4r' - 4r' = 5r'$ and $z = r - x - y \leq 13r' - 3r' - 3r' = 7r'$. But the range $[5r', 7r']$ intersects none of $M$, $N$ and $R$.

\medskip
\emph{Two coordinates in $M$:} If, say, $x, y \in M$, then $z = r - x - y \leq 13r' - 8r' - 8r' < 0$, which clearly contradicts $z \in [0, 13r']$.

\medskip
We conclude that each element of $P$ is contained in $M \times N \times R$. \cref{fig:symmetrization} visualizes this result geometrically: The remaining $6$ regions $\sigma (M \times N \times R)$ for permutations $\sigma \in \S_3$ are precisely the $6$ non-colored triangles, one of which lies in $\overline C$ (the left one at the very bottom). It now follows that the the preimage of $P$ under $\gamma$ is the set $\{ (x - 9r', y - 3r', z) : (x, y, z) \in P \}$, which is a point set in $G_{r'}$ with $(X, Y, Z)$-marginal $(\mu', \nu', \rho')$ as required.

Note that without any changes the whole proof also works for the open cone $C$ instead of $\overline C$, which proves the result for \SkewSymmetric\TwoDXRay.
\end{proof}

% !TEX root = ../paper.tex

\section{Connection to Kronecker Coefficients} \label{sec:kronecker}
\paragraph{Problem 10 in \cite{Stanley1999a}}
Let $V_1$, $V_2$, $V_3$ be finite dimensional complex vector spaces.
Then the group $\GL(V_1) \times \GL(V_2) \times \GL(V_3)$ acts linearly on the tensor product $V_1 \tensor V_2 \tensor V_3$ via $(g_1, g_2, g_3)(v_1 \tensor v_2 \tensor v_3) := (g_1 v_1) \tensor (g_2 v_2) \tensor (g_3 v_3)$ and linear continuation. Hence $\GL(V_1) \times \GL(V_2) \times \GL(V_3)$ also acts linearly on the symmetric power $\Sym^N (V_1 \tensor V_2 \tensor V_3)$, which decomposes into irreducibles
\[
	\Sym^n (V_1 \tensor V_2 \tensor V_3) = \Directsum_{\mu, \nu, \rho} (\Schur^\mu V_1 \tensor \Schur^\nu V_2 \tensor \Schur^\rho V_3)^{\directsum k(\mu, \nu, \rho)}.
\]
These multiplicities $k(\mu, \nu, \rho)$ are called the \emph{Kronecker coefficients} and have been the focus of numerous publications, out of which \cite{BurgisserI2008,BriandOR2009,PakP2017,IkenmeyerMW2017} study their complexity.

Problem~10 in \cite{Stanley1999a}, which is also known as the \emph{Kronecker problem}, asks for a combinatorial description of $k(\mu, \nu, \rho)$ (see Section~\ref{sec:introduction}). Many publications give combinatorial interpretations of the Kronecker coefficients in special subcases \cite{Lascoux1980, Remmel1992, RemmelW1994, Rosas2001, BallantineO2006, Blasiak2012, Hayashi2015, Liu2015, IkenmeyerMW2017} and hence these publications make progress towards resolving the Kronecker problem.

\paragraph{Connecting Problems 9 and 10 in \cite{Stanley1999a}}
Problem 9 in \cite{Stanley1999a} is asking for a combinatorial interpretation for $a_\lambda(n, m)$; this is also known as the \emph{pletyhsm problem}. Much less is known about this question than about the Kronecker problem, although numerous algorithms exist to compute plethysm coefficients (see the references provided in \cite{LoehrR2011}), and \cite{KahleM2018} rules out certain approaches towards finding combinatorial formulas. We make progress on the plethysm problem by giving a combinatorial interpretation in an interesting subcase, see Theorem~\ref{thm:equality-plethysm-kronecker} below.
We carefully mimic the proof technique in \cite{IkenmeyerMW2017}, where a combinatorial interpretation for the Kronecker coefficient is given in a subcase. Since we reach the \emph{same} combinatorial interpretation, we conclude that the plethysm coefficient \emph{equals} the Kronecker coefficient in this case. The detailed statement can be found in Theorem~\ref{thm:equality-plethysm-kronecker} below.

To see the connection, consider the net of reductions as depicted in \cref{fig:reductions}: First, observe that all reductions ultimately start from \Sharp\TwoDXRay{}. Furthermore, any reduction between two tomography problems is parsimonious. In a similar sense, the reductions to computing the Kronecker or plethysm coefficients preserve the number of solutions, too, by guaranteeing that the respective coefficients precisely reflect the number of solutions to the corresponding tomography problem.

For this reason, we can start with an arbitrary \TwoDXRay{} instance $(\mu', \nu', \rho')$ of, say, $N$ solutions. By applying the reductions step-by-step, we obtain coefficients $k(\mu, \nu, \rho)$ and $b_\lambda(n, 3)$ that \emph{both} count the number of solutions to the original instance, hence $k(\mu, \nu, \rho) = N = b_\lambda(n, 3)$. By carefully inspecting the reductions, it is easy to verify the following theorem:

\begin{theorem}\label[theorem]{thm:equality-plethysm-kronecker}
Let $\mu', \nu', \rho' \in \Nat^{[0, r]}$ be compositions satisfying $\sum_{i=0}^r i \mult (\mu_i' + \nu_i' + \rho_i') = r|\mu'| = r|\nu'| = r|\rho'|$ (i.e., the corresponding \TwoDXRay{} instance is not trivially unsatisfiable). Let
\begin{align*}
	\mu := (\mu' + X(Q_{r-1}))^\transpose, &&
	\nu := (\nu' + Y(Q_{r-1}))^\transpose, &&
	\rho := (\rho' + Z(Q_{r-1}))^\transpose,
\end{align*}
where $Q_r := \{ (x, y, z) : x + y + z \leq r \} \subseteq \Nat^3$. Moreover, let
\begin{align*}
	\lambda := \Big((\rho'_0, \ldots, \rho'_{r'}, \underbrace{0, \ldots, 0}_{2r' - 1}, \nu'_0, \ldots, \nu'_{r'}, \underbrace{0, \ldots, 0}_{5r' - 1}, \mu'_0, \ldots, \mu'_{r'}, \underbrace{0, \ldots, 0}_{3r'}) + S(P_{13r-1})\Big)^\transpose,
\intertext{and}
	\pi := (\rho'_0, \ldots, \rho'_{r'}, \underbrace{0, \ldots, 0}_{2r' - 1}, \nu'_0, \ldots, \nu'_{r'}, \underbrace{0, \ldots, 0}_{5r' - 1}, \mu'_0, \ldots, \mu'_{r'}, \underbrace{0, \ldots, 0}_{3r'}) + S(\overline P_{13r-1}).
\end{align*}
Then $k(\mu, \nu, \rho) = a_\lambda(|\lambda| / 3, 3) = b_\pi(|\pi| / 3, 3)$.
\end{theorem}

To illustrate the application of \cref{thm:equality-plethysm-kronecker}, we give some brief examples starting with \TwoDXRay{} instances of size $r = 1$ each. Unfortunately, all such instances are either uniquely satisfiable or unsatisfiable and thus the resulting plethysm and Kronecker coefficients take values $0$ and $1$ only.

\begin{example}
Consider the following \TwoDXRay{} instance $\mu' = (1, 1)$, $\nu' = (1, 1)$ and $\rho' = (2, 0)$ for $r = 1$, which is uniquely satisfiable by $\text{\raisebox{.58ex}[0pt][0pt]{\tikz[triangle diagram, r=1, triangular=.27cm, inline]{\drawtrigrid\drawpoints{1/0/0, 0/1/0}}}} = \{ (1, 0, 0), (0, 1, 0) \}$.

Following \cref{thm:equality-plethysm-kronecker}, let $\mu = ((1, 1) + (3, 1))^\transpose = (4, 2)^\transpose = (2, 2, 1, 1)$, $\nu = (2, 2, 1, 1)$ and $\rho = (2, 1, 1, 1, 1)$. Moreover, since $S(P_{12}) = (30, 26, 22, 19, 16, 13, 11, 9, 6, 4, 2, 1)$, we assign $\lambda = ((2, 0, 0, 1, 1, 0, 0, 0, 0, 1, 1) + S(P_{12}))^\transpose = (12, 11, 11, 10, 10, 9, 8, 8, 8, 7, 7, 6, 6, 5, 5, 5, 5, 4, 4, 4, 3, 3, 2, 2,\allowbreak 2, 2, 1, 1, 1, 1, 1, 1)$. Then $k(\mu, \nu, \rho) = a_\lambda(55, 3) = 1$.
\end{example}

\begin{example}
Consider the \TwoDXRay{} instance $\mu' = (2, 1)$, $\nu' = (2, 1)$ and $\rho' = (2, 1)$ for $r = 1$, which is again uniquely satisfiable by $\text{\raisebox{.58ex}[0pt][0pt]{\tikz[triangle diagram, r=1, triangular=.27cm, inline]{\drawtrigrid\drawpoints{1/0/0, 0/1/0, 0/0/1}}}} = \{ (1, 0, 0), (0, 1, 0), (0, 1, 0) \}$.

By choosing $\mu = (2, 2, 1, 1, 1)$, $\nu = (2, 2, 1, 1, 1)$, $\rho = (2, 2, 1, 1, 1)$ and $\pi = (2, 1, 0, 2, 1, 0, 0, 0,\allowbreak 0, 2, 1) + (61, 54, 46, 38, 31, 23, 17, 12, 9, 6, 4, 2) = (63, 55, 46, 40, 32, 23, 17, 12, 9, 8, 5, 2)$ as in \cref{thm:equality-plethysm-kronecker}, we find that $k(\mu, \nu, \rho) = b_\pi(104, 3) = 1$.
\end{example}

\begin{example}
Consider the \TwoDXRay{} instance $\mu' = (2, 0)$, $\nu' = (2, 0)$ and $\rho' = (0, 2)$ for $r = 1$; it is easy to check that the instance is unsatisfiable.

We construct $\mu = (2, 1, 1, 1, 1)$, $\nu = (2, 1, 1, 1, 1)$, $\rho = (2, 2, 2)$ and $\pi = (0, 2, 0, 2, 0, 0, 0, 0, 0, 2) + (61, 54, 46, 38, 31, 23, 17, 12, 9, 6, 4, 2) = (61, 56, 46, 40, 31, 23, 17, 12, 9, 8, 4, 2)$ as in \cref{thm:equality-plethysm-kronecker}. It follows that $k(\mu, \nu, \rho) = b_\pi(103, 3) = 0$.
\end{example}

% !TEX root = ../paper.tex

\section{Positive Formulas for Restricted Plethysm Coefficients} \label{sec:positive-formulas}
In the course of the main proof, we discovered that a restricted class of plethysm coefficients can be explained as the number of certain combinatorial objects. As motivated before, results of that sort are of independent interest.

Formally speaking, we would like to show that the problem of computing general plethysm coefficients $p_\lambda(\mu, \nu)$ is contained in the complexity class \SharpP{}. Completely resolving that hypothesis seems out of reach with today's techniques. Nevertheless, we make partial progress by proving that for certain nontrivial sets of inputs $(\mu, \nu, \lambda)$, computing $p_\lambda(\mu, \nu)$ actually is in \SharpP{}. In particular, let
\begin{align*}
	\Phi^{\Sym} &:=
	\left\{
		\left(
			\,\tikz[ydiagram, extents=15]{\ycolumnrange{0}{0}{5}{$n$}}\,,
			\,\tikz[ydiagram, extents=31]{\yrow{0}{0}{3}}\,,
			\lambda
		\right) :
		\text{\parbox{6.8cm}{
			\centering$\lambda = \widecheck\lambda + \widehat\lambda \in \Nat^{[0, r]}$ is a partition of $3n$,\\
			where $\widecheck\lambda$ is the sum-marginal of $\overline P_{r-1}$,\\[-.1ex]
			and $\widehat\lambda$ satisfies $\textstyle\sum_{i=0}^r i \mult \widehat\lambda_i = r |\widehat\lambda| / 3$
		}}
	\right\}\!,
\intertext{and}
	\Phi^{\Wedge} &:=
	\left\{
		\left(
			\,\tikz[ydiagram, extents=15]{\ycolumnrange{0}{0}{5}{$n$}}\,,
			\,\tikz[ydiagram, extents=13]{\ycolumn{0}{0}{3}}\,,
			\lambda
		\right) :
		\text{\parbox{6.8cm}{
			\centering$\lambda = \widecheck\lambda + \widehat\lambda \in \Nat^{[0, r]}$ is a partition of $3n$,\\
			where $\widecheck\lambda$ is the sum-marginal of $P_{r-1}$,\\[-.1ex]
			and $\widehat\lambda$ satisfies $\textstyle\sum_{i=0}^r i \mult \widehat\lambda_i = r |\widehat\lambda| / 3$
		}}
	\right\}\!.
\end{align*}
The following theorem follows from~\cref{prop:min-coordinate-sum}.

\begin{theorem} \label[theorem]{thm:sharp-p-column}
The problem of computing plethysm coefficients $p_\lambda(\mu, \nu)$ is in~\SharpP{}, when restricted to instances $(\mu, \nu, \lambda)$ in $\Phi^{\Sym}$ or $\Phi^{\Wedge}$.
\end{theorem}

Using the same idea, we identify a more general class of instances $(\mu, \nu, \lambda)$ for which the same result applies. Namely, for any choice of $\mu$ (and \raisebox{0pt}[0pt][0pt]{$\nu = \,\tikz[ydiagram, extents=31]{\yrow{0}{0}{3}}$} or \raisebox{0pt}[0pt][0pt]{$\nu = \,\tikz[ydiagram, extents=31]{\yrow{0}{0}{3}}\,^\transpose$}), we can construct partitions $\lambda$ in the same spirit as above. Let $\overline p(r) = |\overline P_r|$ denote the number of points in the complete pyramid $\overline P_r$, or equivalently, the total number of partitions of all integers~$\leq r$ into at most three parts. Any Young diagram $\mu = (n_0, n_1, \ldots, n_\ell)^\transpose$ can be written as
\begin{align*}
	\mu =
	\,\begin{tikzpicture}[ydiagram, boxsize={1.2em}{1.2em}, extents={6.5}7]
		\ycolumnrange{0}{0}{7}{$n_0$}{}
		\ycolumnrange{1}{0}{6}{$n_1$}{}
		\ycolumnrange{2}{0}{6}{$n_2$}{}
		\yregion{3}{0}{1.5}{5}
		\ynone{3.25}{.5}{$\cdots\hspace{0em}$}
		\ynone{3.25}{3.5}{\rotatebox{45}{$\cdots\hspace{0em}$}}
		\ycolumnrange{4.5}{0}{4}{$n_\ell$}{}
	\end{tikzpicture}\,
	=
	\,\begin{tikzpicture}[ydiagram, boxsize={1.2em}{1.2em}, extents={6.5}7]
		\ycolumnrange[black!12]{0}{0}{4}{$\widecheck n_0$}{}
		\ycolumnrange[black!12]{1}{0}{3}{$\widecheck n_1$}{}
		\ycolumnrange[black!12]{2}{0}{3}{$\widecheck n_2$}{}
		\yregion[black!12]{3}{0}{1.5}{3}
		\ynone[black!12]{3.25}{.5}{$\cdots\hspace{0em}$}
		\ycolumnrange[black!12]{4.5}{0}{2}{$\widecheck n_\ell$}{}
		\ycolumnrange{0}{4}{3}{$\widehat n_0$}{}
		\ycolumnrange{1}{3}{3}{$\widehat n_1$}{}
		\ycolumnrange{2}{3}{3}{$\widehat n_2$}{}
		\yregion{3}{3}{1.5}{2}
		\ynone{3.25}{3.5}{\rotatebox{45}{$\cdots\hspace{0em}$}}
		\ycolumnrange{4.5}{2}{2}{$\widehat n_\ell$}{}
	\end{tikzpicture}\,
	=
	\,\begin{tikzpicture}[ydiagram, boxsize={1.2em}{1.2em}, extents={6.5}7]
		\draw[box, colored=black!12] (0, 0)
			-- ++(5.5 * \boxx, 0)
			|- ++(-\boxx, 2 * \boxy)
			|- ++(-3.5 * \boxx, \boxy)
			|- ++(-\boxx, \boxy)
			-- cycle;
		\draw[box, colored=white] (0, 4 * \boxy)
			|- ++(\boxx, 3 * \boxy)
			|- ++(2 * \boxx, -\boxy)
			|- ++(1.5 * \boxx, -\boxy)
			|- ++(\boxx, -\boxy)
			|- ++(-\boxx, -2 * \boxy)
			|- ++(-3.5 * \boxx, \boxy)
			|- cycle;
		\node[label=black!12, font=\normalfont] at (.5 * 5.5 * \boxx, 1.5 * \boxy) {$\widecheck\mu$};
		\node[label=white, font=\normalfont] at (.5 * 5.5 * \boxx, 4 * \boxy) {$\widehat\mu$};
	\end{tikzpicture}\,,
\end{align*}
where, for all $0 \leq j \leq \ell$, we pick $r_j$ minimal so that \raisebox{0pt}[0pt][0pt]{$n_j < \overline p(r_j)$} and assign \raisebox{0pt}[0pt][0pt]{$\widecheck n_j := \overline p(r_j - 1)$} and $\widehat n_j := n_j - \widecheck n_j$. Furthermore, let $n := \sum_j n_j = \sum_j \widecheck n_j + \widehat n_j$. As depicted, we refer to the upper rows (that is, rows of small indices) by $\widecheck\mu$ and denote the remaining skew shape by $\widehat\mu$. After decomposing~$\mu$ in that manner, let us define
\begin{align*}
	\Psi^{\Sym} &:=
	\left\{
		\vphantom{\tikz[ydiagram, extents=15]{\ycolumnrange{0}{0}{5}{$n$}}}
		\left(
			\mu,
			\,\tikz[ydiagram, extents=31]{\yrow{0}{0}{3}}\,,
			\lambda
		\right) :
		\text{\parbox{7.8cm}{
			\everymath{\textstyle}
			\centering$\lambda = \sum_j \widecheck\lambda^j + \widehat\lambda^j \in \Nat^{[0, r_0]}$ is a partition of $3n$,\\
			where $\widecheck\lambda^j$ is the sum-marginal of $\overline P_{r_j-1}$,\\[-.3ex]
			and $\widehat\lambda^j$ satisfies $\textstyle\sum_{i=0}^{r_j} i \mult \widehat\lambda^j_i = \widehat n_j \mult r_j$
		}}
	\right\}\!.
\intertext{In a very similar way, we can define $\Psi^{\Wedge}$: This time, let $p(r)$ be the total number of partitions of all integers~$\leq r$ into at most three \emph{distinct} parts (that is, the number of points in the complete pyramid $P_r$). Given $\mu$, determine $r_0, \ldots, r_\ell$ as before with $p(\cdot)$ in place of $\overline p(\cdot)$. Then:}
	\Psi^{\Wedge} &:=
	\left\{
		\vphantom{\tikz[ydiagram, extents=15]{\ycolumnrange{0}{0}{5}{$n$}}}
		\left(
			\mu,
			\,\tikz[ydiagram, extents=13]{\ycolumn{0}{0}{3}}\,,
			\lambda
		\right) :
		\text{\parbox{7.8cm}{
			\everymath{\textstyle}
			\centering$\lambda = \sum_j \widecheck\lambda^j + \widehat\lambda^j \in \Nat^{[0, r_0]}$ is a partition of $3n$,\\
			where $\widecheck\lambda^j$ is the sum-marginal of $P_{r_j-1}$,\\[-.3ex]
			and $\widehat\lambda^j$ satisfies $\textstyle\sum_{i=0}^{r_j} i \mult \widehat\lambda^j_i = \widehat n_j \mult r_j$
		}}
	\right\}\!.
\end{align*}

\begin{theorem} \label[theorem]{thm:sharp-p-arbitrary}
The problem of computing plethysm coefficients $p_\lambda(\mu, \nu)$ is in~\SharpP{}, when restricted to instances $(\mu, \nu, \lambda)$ in $\Psi^{\Sym}$ or $\Psi^{\Wedge}$.
\end{theorem}

Clearly, $\Phi^{\Sym} \subset \Psi^{\Sym}$ and $\Phi^{\Wedge} \subset \Psi^{\Wedge}$, so \cref{thm:sharp-p-arbitrary} really generalizes \cref{thm:sharp-p-column}. For the following proof of \cref{thm:sharp-p-arbitrary}, we will omit the treatment of $\Psi^{\Wedge}$ and focus on $\Psi^{\Sym}$ only; both argumentations are analogous and the only exceptions have been pointed out in the preceding sections.

We catch up on some representation theoretic background. A \emph{Young tableau~$T$} is a Young diagram~$\lambda$ together with a filling of all boxes in $\lambda$ with objects taken from some alphabet. If a total ordering on the alphabet is understood, then we say that~$T$ is \emph{semistandard} whenever the entries of~$T$ are strictly increasing down each column and weakly increasing along each row. The Weyl module~$\Schur^\lambda V$ can explicitly be constructed in terms of Young tableaux: By arbitrarily indexing the boxes of $\lambda$, we fix a basis $\{v_T\}_T$ for \raisebox{0pt}[0pt][0pt]{$\Tensor^{|\lambda|} V$} where $T$ ranges over all Young tableaux of shape $\lambda$ filled with basis elements of $V$. Now $\Schur^\lambda V$ is the largest subspace of \raisebox{0pt}[0pt][0pt]{$\Tensor^{|\lambda|} V$} satisfying the following two exchange conditions (known as the Grassmann-Plücker relations):
\begin{enumerate}[label=\arabic*.]
\item $v_T = -v_{T'}$ where $T'$ is obtained from $T$ by exchanging two vertically adjacent boxes.
\item $v_T = \sum_{T'} v_{T'}$ where, for some $i$ and $\ell$, $T'$ ranges over all tableaux obtained from $T$ by exchanging the top-most $\ell$ boxes in column $i$ with any $\ell$ boxes in column $i - 1$ preserving the vertical order.
\end{enumerate}
It is known that the set of all semistandard tableaux of shape $\lambda$ filled with basis elements of $V$ forms a basis of $\Schur^\lambda V$.

For the remainder of this section, let \raisebox{0pt}[0pt][0pt]{$(\mu, \nu, \lambda) \in \Psi^{\Sym}$} and let $X_0, X_1, \ldots, X_{k-1}$ be an ordered basis of $V = \Complex^k$. Then the set of all monomials $X_x X_y X_z$ for $x, y, z \in [0, k-1]$ forms a basis of $\Sym^3 V$; we sometimes identify monomials $X_x X_y X_z$ with their representatives $(x, y, z) \in \overline C_{< k} := [0, k-1]^3 \intersect \overline C$. Next, let~$\prec$ be the partial ordering on $\overline C_{< k}$ defined by $(x, y, z) \prec (x', y', z')$ whenever $x + y + z < x' + y' + z'$ and arbitrarily extend $\prec$ to a total order. From the previous paragraph we derive that $\{v_T\}_T$ forms a basis of $\Schur^\mu \Schur^\nu V = \Schur^\mu \Sym^3 V$, where $T$ ranges over all semistandard tableaux of shape $\mu$ and alphabet $\overline C_{< k}$. By construction, the weight of a vector $v_T$ equals the sum of all entry-wise weights in $T$.

For a Young tableau $T$ of shape $\mu$, we write $\widecheck T$ to denote the subtableau of shape $\widecheck \mu$ and we write \raisebox{0pt}[0pt][0pt]{$\widehat T$} to denote the skew tableau corresponding to~$\widehat \mu$.

We proceed to rework some parts of \cref{sec:plethysm-bounds,sec:reduction}, where we now consider semistandard tableaux in place of pyramids---the exact same proof is recovered when restricting $\mu$ (and thereby the shape of the tableaux $T$) to a single column.

\begin{lemma} \label[lemma]{lem:tableau-layers}
Let $T$ be a semistandard Young tableau of shape $\mu$ and weight $\lambda$. Then:
\begin{enumerate}[label=\arabic*.]
\item $\widecheck T$ is of weight $\widecheck\lambda := \sum_j \widecheck\lambda^j$. In fact, there exists only one such tableau $\widecheck T$, which is obtained by filling all boxes in the $i$-th row of $\widecheck\mu$ with the $i$-th smallest monomial according to $\prec$.
\item $\widehat T$ is of weight $\widehat\lambda := \sum_j \widehat\lambda^j$. Moreover, the $j$-th column of $\widehat T$ exclusively contains monomials $X_x X_y X_z$ where $x + y + z = r_j$.
\end{enumerate}
\end{lemma}

\noindent
The proof of \cref{lem:tableau-layers} roughly follows \cref{lem:pyramid-embedding}, and to this end we first lift the definition of coordinate sums to tableaux: Let
\begin{align*}
	B(T) = \sum_{(x, y, z)} x + y + z,
\end{align*}
where the sum is over all entries $(x, y, z)$ of $T$. Again, it is easy to relate the coordinate sum of $T$ with its weight: $B(T) = \sum_{i=0}^r i \mult \kappa_i$, where $\kappa \in \Nat^{[0, r]}$ is the weight of $T$.

\begin{proof}[Proof of \cref{lem:tableau-layers}]
Assume that $T$ is a semistandard Young tableau of shape $\mu$ and weight $\lambda$. We determine $B(T)$ exactly:
\begin{align} \label{eq:barycenter-tableau}
	B(T)
	= \sum_{i=0}^{r_0} i \mult \lambda_i
	= \sum_{j=0}^\ell \left(\sum_{i=0}^{r_j} i \mult \widecheck\lambda^j_i + \sum_{i=0}^{r_j} i \mult \widehat\lambda^j_i\!\right)
	= \sum_{j=0}^\ell \left(\sum_{i=0}^{r_j} i \mult \widecheck\lambda^j_i + \widehat n_j \mult r_j\!\right)\!.
\end{align}
Recall that $\widecheck\lambda^j$ equals the sum-marginal of the complete pyramid $\overline P_{r_j}$. In viewing the $j$-th column in $T$ as a point set in the closed cone of size $n_j = |\overline P_{r_j}| + \widehat n_j$, \cref{lem:coordinate-sum-pyramid} yields that the contribution of the $j$-th column to the coordinate sum $B(T)$ is at least \raisebox{0pt}[0pt][0pt]{$\sum_{i=0}^{r_j} i \mult \widecheck\lambda^j_i + \widehat n_j \mult r_j$}. In conjunction with \eqref{eq:barycenter-tableau}, that bound is tight. By applying \cref{lem:coordinate-sum-pyramid} again, we learn that for all $j$, the $j$-th column of $T$ corresponds to a pyramid $P$ that satisfies \raisebox{0pt}[0pt][0pt]{$\overline P_{r_j - 1} \subseteq P \subset \overline P_{r_j}$}. But then the only way to label the boxes in the $j$-th column of $T$ is as stated: The top-most $\widecheck n_j$ boxes are filled with points in \raisebox{0pt}[0pt][0pt]{$\overline P_{r_j - 1}$} and all $\widehat n_j$ boxes below are filled with points in $G_{r_j} \intersect \overline C$.
\end{proof}

\begin{lemma} \label[lemma]{lem:tableau-hwv}
Let $T$ be a semistandard Young tableau of shape $\mu$ and weight $\lambda$. Then $v_T$ is a highest weight vector in $\Schur^\mu \Sym^3 V$.
\end{lemma}
\begin{proof}
We assert that $v_T$ vanishes under all raising operators $E_{a, b}$. In fact, it suffices to show that~$v_T$ vanishes under~$E_{a, b}$ after canonically embedding \raisebox{0pt}[0pt][0pt]{$v_T \in \Schur^\mu \Sym^3 V \subset \Tensor_{j=0}^\ell \Wedge^{n_j} \Sym^3 V$} (that is, the space obtained by omitting condition \textsf{2} in the above characterization of Weyl modules). Let \raisebox{0pt}[0pt][0pt]{$\varphi : \Tensor_{j=0}^\ell \Wedge^{n_j} \Sym^3 V \to \Schur^\mu \Sym^3 V$} be the canonical projection. We rewrite $v_T = \varphi(v_{T_0} \tensor \ldots \tensor v_{T_\ell})$, where each column $T_j$ of $T$ is interpreted as a point set in the closed cone and~$v_P$ for $P$ such a point set is chosen as in \cref{thm:plethysm-bounds}. \cref{lem:tableau-layers} in particular implies that each column~$T_j$ is a pyramid in the closed cone, which entails that $E_{a, b} v_{T_j} = 0$ by \cref{thm:plethysm-bounds}. Hence:
\begin{align*}
	E_{a, b} v_T = E_{a, b} \varphi(v_{T_0} \tensor \ldots \tensor v_{T_\ell}) = \varphi(E_{a, b}v_{T_0} \tensor \ldots \tensor v_{T_\ell}) + \ldots + \varphi(v_{T_0} \tensor \ldots \tensor E_{a, b}v_{T_\ell}) = 0;
\end{align*}
the second equality holds since $\varphi$ is $\GL(V)$-equivariant.
\end{proof}

Recall that $p_\lambda(\mu, \nu)$ equals the dimension of the weight-$\lambda$ highest weight subspace of $\Schur^\mu \Schur^\nu V$. Since each elementary weight-$\lambda$ vector is a highest weight vector in this case, $p_\lambda(\mu, \nu)$ equals the dimension of the weight subspace of weight $\lambda$---or equivalently, the number of semistandard tableaux of shape $\mu$ and weight $\lambda$. Formally, we obtain the following statement which implies \cref{thm:sharp-p-arbitrary} as an immediate consequence.

\begin{lemma}
Let $(\mu, \nu, \lambda) \in \Psi^{\Sym}$. Then $p_\lambda(\mu, \nu)$ equals the number of semistandard tableaux of shape~$\mu$ and weight~$\lambda$, where all boxes are filled with points in $\overline C_{< k}$.
\end{lemma}

% !TEX root = ../paper.tex

\section{GapP-Completeness} \label{sec:gapp-completeness}
This section is devoted to proving~\cref{thm:plethysm-gapp-complete}, which claims that computing general plethysm coefficients $p_\lambda(\mu, \nu)$, and also the special cases $a_\lambda(n, m)$ and $b_\lambda(n, m)$, is \GapP{}-complete. The hardness part follows immediately from~\cref{thm:plethysm-sharpp-hard}\footnotemark. In order to show containment in \GapP{}, we will derive an explicit formula involving signs for the general plethysm coefficients $p_\lambda(\mu, \nu)$. We heavily rely on tools from the theory of symmetric functions, however, we shall refrain from describing the necessary background in detail and instead limit ourselves to the most important definitions; see for instance~\cite{Stanley1999b} for a thorough introduction. The following argument is mostly standard and similar to e.g.~\cite{DorflerIP2019}.
\footnotetext{Indeed, it is folklore that under the appropriate notions of reductions which allow for pre- and post-computations, \SharpP{}-hardness implies \GapP{}-hardness. To see this, let $f$ be \SharpP{}-hard and let $g_1 - g_2 \in \GapP{}$ so that $g_1, g_2 \in \SharpP{}$. The trick is to construct a nonnegative function $g \in \SharpP{}$ by encoding the nonnegative functions $g_1$, $g_2$ in different blocks of bits. Now, since $f$ is \SharpP{}-hard, there exists a polynomial-time reduction from $g$ to $f$ and we can use the post-processing step to recover the intended function value $g_1 - g_2$.}

Let $\Lambda$ denote the ring of symmetric polynomials (that is, polynomials which are invariant under any permutation of variables) in finitely many variables $X_0, X_1, \ldots$. The characters of the irreducible $\GL_n$-representation $\Schur^\nu \Complex^n$ are elements $s_\nu(X_0, \ldots, X_{n-1}) \in \Lambda$ called \emph{Schur functions} (where $X_0, \ldots, X_{n-1}$ correspond to the eigenvalues of the conjugacy class of $\GL_n$). Recall the notion of \emph{semistandard Young tableaux} as introduced in \cref{sec:positive-formulas} and let $T_0, \ldots, T_{l-1}$ be an arbitrary enumeration of all semistandard tableaux of shape $\nu$ filled with entries $0, \ldots, n-1$. The \emph{weight $\weight(T) \in \Nat^{[0, n-1]}$} of a tableau $T$ is defined as the composition such that $\weight(T)_i$ is the number of entries $i$ in $T$. We have that
\begin{align}
	s_\nu(X_0, \ldots, X_{n-1}) = \sum_{i=0}^{l-1} X^{T_i}, \label{eq:schur-functions}
\end{align}
for $X^T := \prod_{t \in T} X_t$, where the product is over (the multiset of) the entries $t$ of $T$. Following this description, it follows that the character of the composed representation $\Schur^\mu \Schur^\nu \Complex^n$ is given by
\begin{align*}
	s_\mu[s_\nu](X_0, \ldots, X_{n-1}) := s_\mu(X^{T_0}, \ldots, X^{T_{l-1}});
\end{align*}
that composition $s_\mu[s_\nu]$ of symmetric functions is called a \emph{symmetric function plethysm}.

The set of Schur function forms a basis of $\Lambda$. Moreover, since the Schur functions are the characters of the irreducible $\GL_n$-representations, there exists an inner product $\innerprod\cdot\cdot$ on $\Lambda$ so that the Schur functions form an orthonormal basis: $\innerprod{s_\lambda}{s_\mu} = \delta_{\lambda \mu}$. It follows that for any $\GL_n$-representation with character $\chi$, the multiplicity of $\Schur^\lambda \Complex^n$ in its decomposition into irreducibles equals $\innerprod{s_\lambda}{\chi}$. In particular, $p_\lambda(\mu, \nu) = \innerprod{s_\lambda}{s_\mu[s_\nu]}$. It remains to develop a \GapP{}-formula for $\innerprod{s_\lambda}{s_\mu[s_\nu]}$.

As a special case of Schur functions, the \emph{complete homogeneous symmetric functions $h_d$} are defined as $s_{(d)}$; we write $h_\lambda := h_{\lambda_0} h_{\lambda_1} \cdots$. Moreover, let $m_\lambda := \sum_{\rho \in \S_n(\lambda)} X^\rho$ denote the \emph{monomial symmetric functions}, where $\S_n(\lambda)$ is the set of all (distinct) permutations of $\lambda = (\lambda_0, \ldots, \lambda_{n-1})$ and $X^\rho := \prod_i X_i^{\rho_i}$. The well-known Jacobi-Trudi identity~\cite[Section~7.16]{Stanley1999b} expresses any Schur function as a determinant:
\begin{align}
	s_\lambda = \det (h_{\lambda_i - i + j})_{i, j = 0}^{\ell(\lambda) - 1}, \label{eq:jacobi-trudi}
\end{align}
thus $s_\lambda$ can be written as a signed sum of polynomials $h_\mu$. In addition, it is known that the orthogonal dual basis for the complete homogeneous symmetric functions is given by the monomial symmetric functions, i.e.~$\innerprod{h_\lambda}{m_\mu} = \delta_{\lambda \mu}$. Now suppose that $s_\mu[s_\nu]$ can be expressed as
\begin{align*}
	s_\mu[s_\nu] = \sum_\kappa q_\kappa(\mu, \nu) m_\kappa,
\end{align*}
for some \GapP{}-computable coefficients $q_\kappa(\mu, \nu)$. By the Jacobi-Trudi identity~\eqref{eq:jacobi-trudi}, we have
\begin{align*}
	p_\lambda(\mu, \nu) = \innerprod{\det (h_{\lambda_i - i + j})_{i, j = 0}^{\ell(\lambda)-1}}{s_\mu[s_\nu]} = \sum_{\sigma \in \S_{\ell(\lambda)}} \sign(\sigma) q_{\lambda - (0, \ldots, \ell(\lambda)-1) + \sigma}(\mu, \nu),
\end{align*}
where the permutation $\sigma$ is viewed as a vector with entries $0, \ldots, \ell(\lambda) - 1$. This formula implies in particular that also $p_\lambda(\mu, \nu)$ can be computed by a \GapP{} machine.

It now suffices to prove that the coefficients $q_\kappa(\mu, \nu)$ are indeed expressible as a \GapP{} formula. From~\eqref{eq:schur-functions}, one can derive that
\begin{align*}
	s_\mu = \sum_\pi K_{\mu \pi} m_\pi,
\end{align*}
where $K_{\mu \pi}$ denotes the number of semistandard tableaux of shape $\mu$ and weight $\pi$ (called the \emph{Kostka coefficient}). We finally obtain:
\begin{align*}
	s_\mu[s_\nu](X_0, \ldots, X_{n-1})
	&= s_\mu(X^{T_0}, \ldots, X^{T_{l-1}}) \\
	&= \sum_\pi K_{\mu \pi} m_\pi(X^{T_0}, \ldots, X^{T_{l-1}}) \\
	&= \sum_\pi K_{\mu \pi} \sum_{\rho \in \S_l(\pi)} X^{\sum_i \rho_i \weight(T_i)},
\end{align*}
and therefore
\begin{align*}
	q_\kappa(\mu, \nu) = \innerprod{h_\kappa}{s_\mu[s_\nu]} = \sum_\pi K_{\mu \pi} \mult \#\{ \rho \in \S_l(\pi) : \kappa = \sum_i \rho_i \weight(T_i) \}.
\end{align*}
From that identity it is obvious that $q_\kappa(\mu, \nu)$ is \GapP{}-computable (even \SharpP{}-computable): The coefficients $K_{\mu \pi}$ are clearly \SharpP{}-computable and after guessing $\pi$ and $\rho$, we can check the condition $\kappa = \sum_i \rho_i \weight(T_i)$ in polynomial time.

\bibliographystyle{alpha}
\bibliography{references}

\begin{thebibliography}{CDKW14}

\bibitem[BDLG01]{BrunettiDG2001}
Sara Brunetti, Alberto Del~Lungo, and Yan Gérard.
\newblock On the computational complexity of reconstructing three-dimensional
  lattice sets from their two-dimensional {X}-rays.
\newblock {\em Linear Algebra and its Applications}, 339(1):59--73, 2001.

\bibitem[BHI17]{BurgisserHI2017}
Peter B{\"u}rgisser, Jesko H{\"u}ttenhain, and Christian Ikenmeyer.
\newblock Permanent versus determinant: {Not} via saturations.
\newblock {\em Proceedings of the American Mathematical Society},
  145:1247--1258, 2017.

\bibitem[BI08]{BurgisserI2008}
Peter B{\"u}rgisser and Christian Ikenmeyer.
\newblock The complexity of computing {Kronecker} coefficients.
\newblock In {\em FPSAC 2008, Valparaiso-Vi{\~n}a del Mar, Chile, DMTCS proc.
  AJ}, pages 357--368, 2008.

\bibitem[BI11]{BurgisserI2011}
Peter B{\"u}rgisser and Christian Ikenmeyer.
\newblock Geometric complexity theory and tensor rank.
\newblock In {\em Proceedings of the 43rd Annual ACM Symposium on Theory of
  Computing}, STOC 2011, pages 509--518, New York, NY, USA, 2011. ACM.

\bibitem[BI13a]{BurgisserI2013a}
Peter B{\"u}rgisser and Christian Ikenmeyer.
\newblock Deciding positivity of {Littlewood}-{Richardson} coefficients.
\newblock {\em SIAM Journal of Discrete Mathematics}, 27:1639--1681, 2013.

\bibitem[BI13b]{BurgisserI2013b}
Peter B\"{u}rgisser and Christian Ikenmeyer.
\newblock Explicit lower bounds via geometric complexity theory.
\newblock In {\em Proceedings of the 45th Annual ACM Symposium on Theory of
  Computing}, STOC 2013, pages 141--150, New York, NY, USA, 2013. ACM.

\bibitem[BI18]{BlaserI2018}
Markus Bl{\"a}ser and Christian Ikenmeyer.
\newblock Introduction to geometric complexity theory.
\newblock Lecture notes,
  \url{http://pcwww.liv.ac.uk/~iken/teaching_sb/summer17/introtogct/gct.pdf},
  version from July 25, 2018.

\bibitem[BIP19]{BurgisserIP2019}
Peter B{\"u}rgisser, Christian Ikenmeyer, and Greta Panova.
\newblock No occurrence obstructions in geometric complexity theory.
\newblock {\em Journal of the American Mathematical Society}, 32:163--193,
  2019.

\bibitem[Bla12]{Blasiak2012}
Jonah Blasiak.
\newblock {Kronecker} coefficients for one hook shape.
\newblock arXiv:1209.2018, 2012.

\bibitem[BLMW11]{BurgisserLMW2011}
Peter B\"urgisser, Joseph~M. Landsberg, Laurent Manivel, and Jerzy Weyman.
\newblock An overview of mathematical issues arising in the geometric
  complexity theory approach to {VP} v.s. {VNP}.
\newblock {\em SIAM J. Comput.}, 40(4):1179--1209, 2011.

\bibitem[BO05]{BallantineO2006}
Cristina~M. Ballantine and Rosa~C. Orellana.
\newblock A combinatorial interpretation for the coefficients in the
  {Kronecker} product {$s\sb {(n-p,p)}\ast s\sb \lambda$}.
\newblock {\em S\'em. Lothar. Combin.}, 54A:Art. B54Af, 29 pp. (electronic),
  2005.

\bibitem[BOR09]{BriandOR2009}
Emmanuel Briand, Rosa Orellana, and Mercedes~H. Rosas.
\newblock Reduced {Kronecker} coefficients and counter-examples to {Mulmuley}'s
  strong saturation conjecture {SH}.
\newblock {\em Computational Complexity}, 18(4):577, 2009.

\bibitem[B{\"u}r16]{Burgisser2016}
Peter B{\"u}rgisser.
\newblock Permanent versus determinant, obstructions, and {Kronecker}
  coefficients.
\newblock {\em S\'eminaire Lotharingien de Combinatoire}, 75:1--19, 2016.
\newblock Article B75a.

\bibitem[Car90]{Carre1990}
Christophe Carré.
\newblock Plethysm of elementary functions.
\newblock {\em Bayreuther Mathematische Schriften}, 31:1--18, 1990.

\bibitem[CDKW14]{ChristandlDKW2014}
Matthias Christandl, Brent Doran, Stavros Kousidis, and Michael Walter.
\newblock Eigenvalue distributions of reduced density matrices.
\newblock {\em Communications in Mathematical Physics}, 332(1):1--52, 2014.

\bibitem[CHM07]{ChristiandlHM2007}
Matthias Christandl, Aram~W. Harrow, and Graeme Mitchison.
\newblock Nonzero {K}ronecker coefficients and what they tell us about spectra.
\newblock {\em Communications in Mathematical Physics}, 270:575--585, 2007.

\bibitem[DIP19]{DorflerIP2019}
Julian D\"orfler, Christian Ikenmeyer, and Greta Panova.
\newblock On geometric complexity theory: {Multiplicity} obstructions are
  stronger than occurrence obstructions.
\newblock In {\em Proceedings of the 46th International Colloquium on Automata,
  Languages, and Programming}, ICALP 2019, pages 51:1--51:14, 2019.

\bibitem[DLM06]{DeLoeraM2006}
Jes{\'u}s~A. De~Loera and Tyrrell~B. McAllister.
\newblock On the computation of {Clebsch}-{Gordan} coefficients and the
  dilation effect.
\newblock 15(1):7--19, 2006.

\bibitem[FFK94]{FennerFK1994}
Stephen~A. Fenner, Lance~J. Fortnow, and Stuart~A. Kurtz.
\newblock Gap-definable counting classes.
\newblock {\em Journal of Computer and System Sciences}, 48(1):116--148, 1994.

\bibitem[FH91]{FultonH1991}
William Fulton and Joe Harris.
\newblock {\em Representation Theory -- {A} First Course}, volume 129 of {\em
  Graduate Texts in Mathematics}.
\newblock Springer, 1991.

\bibitem[GGP99]{GardnerGP1999}
Richard~J. Gardner, Peter Gritzmann, and Dieter Prangenberg.
\newblock On the computational complexity of reconstructing lattice sets from
  their {X}-rays.
\newblock {\em Discrete Math.}, 202(1-3):45--71, 1999.

\bibitem[GIP17]{GesmundoIP2017}
Fulvio Gesmundo, Christian Ikenmeyer, and Greta Panova.
\newblock Geometric complexity theory and matrix powering.
\newblock {\em Differential Geometry and its Applications}, 55:106--127, 2017.
\newblock Part of special issue: Geometry and complexity theory.

\bibitem[Hay15]{Hayashi2015}
Takahiro Hayashi.
\newblock A decomposition rule for certain tensor product representations of
  the symmetric groups.
\newblock arXiv:1507.02047, 2015.

\bibitem[IK19]{IkenmeyerK2019}
Christian Ikenmeyer and Umangathan Kandasamy.
\newblock Implementing geometric complexity theory: {On} the separation of
  orbit closures via symmetries.
\newblock arXiv:1911.03990, 2019.

\bibitem[Ike12]{Ikenmeyer2012}
Christian Ikenmeyer.
\newblock {\em Geometric Complexity Theory, Tensor Rank, and
  {Littlewood}-{Richardson} Coefficients}.
\newblock PhD thesis, Institute of Mathematics, University of Paderborn, 2012.
\newblock Online available at
  \url{http://nbn-resolving.de/urn:nbn:de:hbz:466:2-10472}.

\bibitem[Ike19]{Ikenmeyer2019}
Christian Ikenmeyer.
\newblock {GCT} and symmetries.
\newblock
  \url{http://pcwww.liv.ac.uk/~iken/teaching_sb/winter1718/gct2/symmetries.pdf},
  version from January 30, 2019.

\bibitem[IMW17]{IkenmeyerMW2017}
Christian Ikenmeyer, Ketan~D. Mulmuley, and Michael Walter.
\newblock On vanishing of {Kronecker} coefficients.
\newblock {\em Comput. Complex.}, 26(4):949--992, 2017.

\bibitem[IP17]{IkenmeyerP2017}
Christian Ikenmeyer and Greta Panova.
\newblock Rectangular {Kronecker} coefficients and plethysms in geometric
  complexity theory.
\newblock {\em Advances in Mathematics}, 319:40--66, 2017.

\bibitem[KL14]{KadishL2012}
Harlan Kadish and Joseph~M. Landsberg.
\newblock Padded polynomials, their cousins, and geometric complexity theory.
\newblock {\em Communications in Algebra}, 42(5):2171--2180, 2014.

\bibitem[KL19]{KimotoS2019}
Kazufumi Kimoto and Soo~Teck Lee.
\newblock Highest weight vectors in plethysms.
\newblock {\em Communications in Mathematical Physics}, Dec 2019.

\bibitem[KM16a]{KahleM2016}
Thomas Kahle and Mateusz Micha\l{}ek.
\newblock Plethysm and lattice point counting.
\newblock {\em Foundations of Computational Mathematics}, 2016.

\bibitem[KM16b]{KahleM2015}
Thomas Kahle and Mateusz Micha\l{}ek.
\newblock Plethysm and lattice point counting.
\newblock arXiv:1408.5708v3. This preprint version extends the version in
  \emph{Foundations of Computational Mathematics}, 2016.

\bibitem[KM18]{KahleM2018}
Thomas Kahle and Mateusz Micha\l{}ek.
\newblock Obstructions to combinatorial formulas for plethysm.
\newblock {\em The Electronic Journal of Combinatorics}, 25, 2018.

\bibitem[Kra85]{Kraft1985}
Hanspeter Kraft.
\newblock {\em Geometrische {Methoden} in der {Invariantentheorie}}.
\newblock Friedr. Vieweg und Sohn Verlagsgesellschaft, Braunschweig, 1985.

\bibitem[Kum15]{Kumar2015}
Shrawan Kumar.
\newblock A study of the representations supported by the orbit closure of the
  determinant.
\newblock {\em Compositio Mathematica}, 151(2):292--312, 2015.

\bibitem[Lan11]{Landsberg2011}
Joseph Landsberg.
\newblock {\em Tensors: {Geometry} and Applications}, volume 128 of {\em
  Graduate Studies in Mathematics}.
\newblock American Mathematical Society, Providence, Rhode Island, 2011.

\bibitem[Lan17]{Landsberg2017}
Joseph~M. Landsberg.
\newblock {\em Geometry and Complexity Theory}.
\newblock Cambridge Studies in Advanced Mathematics. Cambridge University
  Press, 2017.

\bibitem[Las80]{Lascoux1980}
Alain Lascoux.
\newblock Produit de {Kronecker} des repr\'esentations du groupe sym\'etrique.
\newblock In {\em S\'eminaire d'Alg\`ebre Paul Dubreil et Marie-Paule
  Malliavin, 32\`eme ann\'ee (Paris, 1979)}, volume 795 of {\em Lecture Notes
  in Math.}, pages 319--329. Springer, Berlin, 1980.

\bibitem[Liu17]{Liu2015}
Ricky~I. Liu.
\newblock A simplified {Kronecker} rule for one hook shape.
\newblock {\em Proc. Amer. Math. Soc.}, 145:3657--3664, 2017.

\bibitem[LR11]{LoehrR2011}
Nicholas Loehr and Jeffrey Remmel.
\newblock A computational and combinatorial expos\'e of plethystic calculus.
\newblock {\em Journal of Algebraic Combinatorics}, 33:163--198, 2011.

\bibitem[Man98]{Manivel1998}
Laurent Manivel.
\newblock Gaussian maps and plethysm.
\newblock {\em Algebraic Geometry (Catania, 1993 and Barcelone, 1994), Lecture
  Notes in Pure and Appl. Math.}, 200:91--117, 1998.

\bibitem[Man11]{Manivel2011}
Laurent Manivel.
\newblock On rectangular {Kronecker} coefficients.
\newblock {\em Journal of Algebraic Combinatorics}, 33(1):153--162, 2011.

\bibitem[MM15]{ManivelM2015}
Laurent Manivel and Mateusz Micha\l{}ek.
\newblock Secants of minuscule and cominuscule minimal orbits.
\newblock {\em Linear Algebra and its Applications}, 481:288--312, 2015.

\bibitem[MNS12]{MulmuleyNS2012}
Ketan~D. Mulmuley, Hariharan Narayanan, and Milind Sohoni.
\newblock Geometric complexity theory~{III}: {On} deciding nonvanishing of a
  {Littlewood}-{Richardson} coefficient.
\newblock {\em Journal of Algebraic Combinatorics}, 36(1):103--110, 2012.

\bibitem[MS01]{MulmuleyS2001}
Ketan~D. Mulmuley and Milind Sohoni.
\newblock Geometric complexity theory~{I}: {An} approach to the {P} vs.\ {NP}
  and related problems.
\newblock {\em SIAM J. Comput.}, 31(2):496--526, 2001.

\bibitem[MS08]{MulmuleyS2008}
Ketan~D. Mulmuley and Milind Sohoni.
\newblock Geometric complexity theory~{II}: {Towards} explicit obstructions for
  embeddings among class varieties.
\newblock {\em SIAM J. Comput.}, 38(3):1175--1206, 2008.

\bibitem[Mul07]{Mulmuley2007}
Ketan~D. Mulmuley.
\newblock Geometric complexity theory~{VII}: {Nonstandard} quantum group for
  the plethysm problem.
\newblock arXiv:0709.0749, 2007.

\bibitem[Mul09]{Mulmuley2009}
Ketan~D. Mulmuley.
\newblock Geometric complexity theory~{VI}: {The} flip via saturated and
  positive integer programming in representation theory and algebraic geometry.
\newblock arXiv:0704.0229v4, 2009.

\bibitem[Nar06]{Narayanan2006}
Hariharan Narayanan.
\newblock On the complexity of computing {Kostka} numbers and
  {Littlewood}-{Richardson} coefficients.
\newblock {\em Journal of Algebraic Combinatorics}, 24(3):347--354, 2006.

\bibitem[PP17]{PakP2017}
Igor Pak and Greta Panova.
\newblock On the complexity of computing {Kronecker} coefficients.
\newblock {\em Comput. Complex.}, 26(1):1--36, 2017.

\bibitem[Rem92]{Remmel1992}
Jeffrey~B. Remmel.
\newblock Formulas for the expansion of the {Kronecker} products {$S\sb
  {(m,n)}\otimes S\sb {(1\sp {p-r},r)}$} and {$S\sb {(1\sp k2\sp l)}\otimes
  S\sb {(1\sp {p-r},r)}$}.
\newblock {\em Discrete Math.}, 99(1-3):265--287, 1992.

\bibitem[Ros01]{Rosas2001}
Mercedes~H. Rosas.
\newblock The {Kronecker} product of {Schur} functions indexed by two-row
  shapes or hook shapes.
\newblock {\em J. Algebraic Combin.}, 14(2):153--173, 2001.

\bibitem[RW94]{RemmelW1994}
Jeffrey~B. Remmel and Tamsen Whitehead.
\newblock On the {Kronecker} product of {Schur} functions of two row shapes.
\newblock {\em Bull. Belg. Math. Soc. Simon Stevin}, 1(5):649--683, 1994.

\bibitem[Sag01]{Sagan2001}
Bruce~E. Sagan.
\newblock {\em The symmetric group: {Representations}, combinatorial
  algorithms, and symmetric functions}, volume 203 of {\em Graduate Texts in
  Mathematics}.
\newblock Springer, New York, second edition, 2001.

\bibitem[Sta99a]{Stanley1999b}
Richard~P. Stanley.
\newblock {\em Enumerative Combinatorics: {Volume}~2}.
\newblock Cambridge Studies in Advanced Mathematics. Cambridge University
  Press, 1999.

\bibitem[Sta99b]{Stanley1999a}
Richard~P. Stanley.
\newblock Positivity problems and conjectures in algebraic combinatorics.
\newblock In {\em Mathematics: Frontiers and Perspectives}, pages 295--319.
  American Mathematical Society, 1999.

\bibitem[TC92]{Thibon1992}
Jean-Yves Thibon and Christophe Carr\'e.
\newblock Plethysm and vertex operators.
\newblock {\em Advances in applied Mathematics}, 13:390--403, 1992.

\bibitem[Wei90]{Weintraub1990}
Steven~H. Weintraub.
\newblock Some observations on plethysms.
\newblock {\em Journal of Algebra}, 129:103--114, 1990.

\bibitem[Wey03]{Weyman2003}
Jerzy Weyman.
\newblock {\em Cohomology of Vector Bundles and Syzygies}.
\newblock Cambridge Tracts in Mathematics. Cambridge University Press, 2003.

\end{thebibliography}

\end{document}